\documentclass[12pt]{article}

\usepackage{amsmath}
\usepackage{mathrsfs}
\usepackage{amsfonts}
\usepackage{graphicx} 
\usepackage{setspace} 
\usepackage{titletoc} 
\usepackage{amssymb}
\usepackage{setspace}
\usepackage{lineno}
\usepackage{amsthm}
\usepackage{booktabs}
\usepackage{xcolor}
\usepackage{soul}
\newtheorem{theorem}{Theorem}

\usepackage[top=2.54cm, bottom=3.81cm, left=2.54cm, right=2.54cm]{geometry}

\usepackage{multirow} 
\usepackage{listings} 
\usepackage{pdfpages} 
\usepackage[hidelinks]{hyperref} 

\usepackage[backend=bibtex,style=numeric,sorting=none]{biblatex}
\renewbibmacro{in:}{}

\title{A Combined ODE Model of Carbohydrate Fermentation and Colorectal Cancer}
\author{A. Lawryshyn, H.J. Eberl, T. Hillen }
\date{}

\begin{document}

\maketitle

\pagenumbering{arabic} \setcounter{page}{1}

\begin{abstract}
    We formulate and analyze a system of non-linear ordinary differential equations that describe key metabolic and immunological interactions between butyrate produced by fiber-fermenting gut microbiota, colorectal cancer cells and host cell populations. The model is studied both independently and in conjunction with a pre-existing carbohydrate fermentation model. The parameter space is explored through sensitivity analyses. Simulation experiments are conducted to illustrate the emergence of varying dynamical behaviour driven by butyrate availability. Our model predicts that butyrate production is driven by fiber consumption and further supported by probiotics in the case of microbial dysbiosis. It also suggests that butyrate may help in suppressing tumour growth. We also show that by adding noise with sufficiently high intensity, cancer elimination occurs almost surely in infinite time and that this threshold level of noise intensity decreases with increasing butyrate concentrations.
\end{abstract}

\section{Introduction}
\subsection{The Human Colon and Gut Microbiome}
The human colon is host to a diverse population of microorganisms, including bacteria, fungi, and archaea, that are collectively referred to as gut microbiota. These microbiota play an essential role in digestion, metabolism, and immune system support. The gut microbiome has become a popular research topic over the past few decades, with increasing evidence highlighting its connection to a wide range of diseases including inflammatory bowel diseases (IBD), irritable bowel syndrome (IBS), obesity, diabetes and cancer \cite{jandhyala2015role, kim2022potential}. 
The specific composition of the gut microbiome is unique to each individual, depending on factors such as genetics, diet, and antibiotic use \cite{jandhyala2015role}. 

Fiber fermenting bacteria, including species from the \textit{Bacterioides} and \textit{Bifidobacterium} genera in addition to others \cite{oliphant2019macronutrient}, are responsible for the fermentation of complex carbohydrates into short-chain fatty acids (SCFAs). Acetate, propionate and butyrate are the most abundant SCFAs in the gut. They are essential to human health, playing a significant role in energy metabolism, gut barrier strength and immune response modulation \cite{hou2022gut}. 

\subsubsection{Carbohydrate Fermentation}
The human colon is lined with a two-layer mucus membrane: an inner layer that adheres to the epithelial cells and is devoid of bacteria, and an outer, unattached layer which is colonized by microbiota \cite{johansson2011two}. The lumen of the colon is also colonized but encompasses a different microbial composition than that of the mucus layer.

Glycans (dietary fiber and mucins) are the main energy sources for carbohydrate fermenting gut microbiota. The majority of dietary fibers originate from plant cell walls, obtained by the consumption of fruits, vegetables and whole grains \cite{williams2017gut}. Dietary fiber escapes digestion in the upper gastrointestinal (GI) tract and enters the lumen of the colon through the small intestine. Conversely, mucins, which are glycoproteins produced by the host's epithelial tissue, are abundant in the mucus layer. Only certain microbial species are capable of degrading fiber and mucins into simpler sugars. These sugars are further broken down by other species into SCFAs, lactate, and gas.

\subsection{The Gut Microbiome and Colorectal Cancer}
Colorectal cancer (CRC) is the third most common type of cancer and the second leading cause of cancer-related deaths worldwide \cite{who_crc_stats}. Growing evidence highlights important connections between CRC and the gut microbiome \cite{kim2022potential, fulbright2017microbiome, wong2023gut, hibberd2017intestinal, }. In particular, it has been shown that butyrate, produced from fiber-fermenters in the gut, has an inhibitory effect on cancer cell growth but promotes that of healthy cells. This conflicting effect of butyrate, known as the ``butyrate paradox", is thought to be related to the contrasting metabolic behaviour of normal versus cancerous cells as well as the supportive role of butyrate on the immune system \cite{li2018butyrate, sun2024butyrate}. 

\subsubsection{Metabolic role of butyrate in cancer growth}
Butyrate is present in high concentrations in the human colon, and is utilized by colonocytes as their primary energy source, accounting for approximately 70\% of their total energy requirements \cite{li2018butyrate}. Unlike colonocytes, which preferentially metabolize butyrate via beta-oxidation, cancer cells undergo the Warburg effect \cite{warburg1956origin}, relying on glucose as an energy source, which favours their fast growth and survival \cite{jacquet2022searching}. Butyrate has been shown to inhibit this effect by acting as a histone deacetylase (HDAC) inhibitor, shifting cancer cell metabolism toward respiration and thereby slowing proliferation \cite{li2018butyrate}. 

\subsubsection{Butyrate affects immune response to cancer}
The immune system's response to cancer is highly complex, involving many different biochemical pathways. We focus on a high level overview, incorporating three key types of immune cells in our model - cytotoxic T cells, regulatory T cells (Tregs) and dendritic cells.

Dendritic cells capture tumour antigens and present them to T cells, causing their activation and recruitment to the tumour microenvironment, initiating an immune response \cite{marciscano2021role}. T cells can kill the cancer cells by releasing molecules that induce apoptosis \cite{ahmed2023role}. A high presence of T cells is often associated with better control of cancer progression, although certain tumours can evade destruction by impairing T cell infiltration and function \cite{pu2025t}. Tregs are important for preventing autoimmunity but in the presence of cancer, they can directly inhibit T-cells, creating an immunosuppressive environment that allows for cancer growth \cite{ohue2019regulatory}. 

Through its role as a HDAC inhibitor, butyrate increases the expression of the protein Foxp3 which is the primary transcription factor required for Treg development and function \cite{kespohl2017microbial}. Butyrate also boosts T cell function and proliferation \cite{danne2021butyrate}. Figure \ref{fig: immune cell diagram} shows a visual depiction of the key stimulatory and inhibitory interactions between butyrate, cancer cells, colonocytes, and immune cells. 

\begin{figure}
    \centering
    \includegraphics[width=1.0\linewidth]{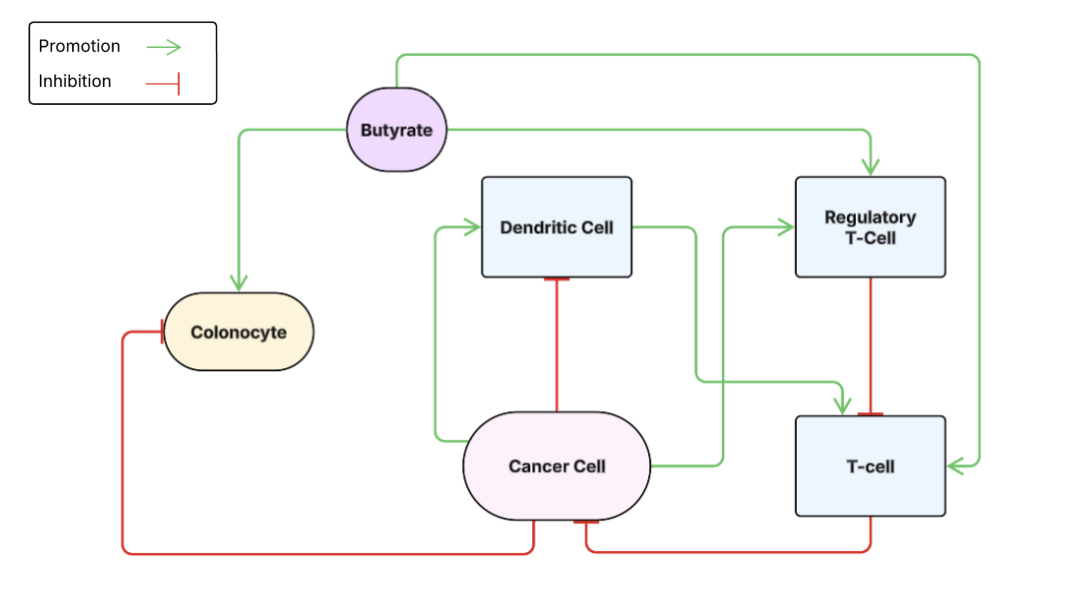}
    \caption[Diagram of immune system interactions]{Diagram of primary interactions of butyrate, cancer and normal cells with the immune system.}
    \label{fig: immune cell diagram}
\end{figure}

\subsection{Effect of Probiotic Administration on Butyrate Production and CRC Progression}

\subsection{Role of Fiber in CRC Immunotherapy}\label{subsect: immunotherapy}
Immunotherapy is a targeted treatment that helps strengthen the immune system's response to cancer. It has transformed treatment strategies against many types of cancers but immunotherapy efficacy against CRC is currently limited \cite{brahmer2012safety}. Tumours with high microsatellite instability (MSI) are often treated with Programmed Cell Death Protein 1 (PD-1) inhibitors. PD-1 is a protein on T cells that, under normal conditions, help prevent them from attacking other cells in the body. However, only 10\%-15\% of total CRC cases have high MSI, with the rest being microsatellite stable (MSS) \cite{boland2010microsatellite}. Microsatellites are short, repetitive sequences of DNA distributed throughout the genome. These regions are particularly susceptible to errors during DNA replication, most of which are corrected by the DNA mismatch repair (MMR) system. However, when the MMR system is impaired, these errors accumulate and can lead to alterations in microsatellite sequences, causing MSI. Butyrate has been shown to improve efficacy of PD-1 inhibitors in both MSI-high and MSS tumours, suggesting that increased fiber intake may be beneficial for CRC patients who are resistant to PD-1 inhibitors \cite{danne2021butyrate, kang2023roseburia}.

\subsection{Cancer Modelling}
Mathematical oncology has gained traction over the past several decades, providing a theoretical framework for clinical studies \cite{rockne2019introduction}. In-silico experiments are now prevalent in the medical field, as in-vitro and in-vivo experiments are often expensive, time-consuming, or ethically challenging. In this thesis, we construct an ordinary differential equation (ODE) model describing interactions between butyrate produced by gut microbiota and colon cancer. Several CRC models have been developed that incorporate gut dynamics. The authors in \cite{johnston2007mathematical} constructed a compartmental model of CRC cell population dynamics in the colonic crypt, describing interactions between stem cells, differentiated cells and transit cells. In \cite{lo2013mathematical}, a reaction-diffusion model of colitis-associated colon cancer was formulated to study the impact of chronic inflammation of the colon mucosa on tumour suppressing genes. A model of CRC initiation in the colonic crypt was introduced in \cite{paterson2020mathematical} and a reaction-diffusion system of colon cancer cell metabolic reprogramming was studied in \cite{lee2017mathematical}. 

Our model is constructed to focus on cancer-butyrate dynamics, which primarily involves interactions with specific immune cells --- T cells, Tregs and dendritic cells, but many cancer models have been developed that include different combinations of immune elements. For example, \cite{alharbi2020new} and \cite{de2001mathematical} formulate normal and cancer growth as competition models, both including immune cells as a single component. The model in \cite{raeisi2024mathematical} includes colon cancer, dendritic, helper T and cytotoxic T cells. Additionally, \cite{mohammad2022investigating} and \cite{kirshtein2020data} constructed detailed models for breast and colon cancer respectively, each incorporating various types of immune cells (including T cells, Tregs and dendritic cells) and cytokines. In \cite{hadjigeorgiou2025mathematical}, a large ODE system was constructed that incorporates interactions between colon cancer cells, immune cells, and various families of gut microbiota. A comprehensive review of cancer-immune cell modelling is presented in \cite{eftimie2011interactions}. Butyrate was incorporated in an immunological model describing inflammation in the human gut \cite{neumann2018qualitative}, as well as in a model of the gut-bone axis \cite{islam2021mathematical}. However, to the best of our knowledge, no existing mathematical models describe the role of butyrate as a cellular metabolite in CRC.

\subsection{Incorporating Stochastic Cancer Dynamics}
Mathematical models are powerful tools for studying biological systems, but they necessarily involve simplifying assumptions. One way to address this limitation is to incorporate randomness into the equations, transforming a deterministic ODE system into a system of stochastic differential equations (SDEs). This helps to capture some of the uncertainty involved in biological processes such as cancer growth \cite{allen2010introduction}. Several existing cancer models incorporate SDEs to represent effects such as cell mutation, division, and micro-environmental interactions \cite{mao2011stationary, mansour2022stochastic, tabassum2019mathematical, azizi2026advancing, shrestha2025coupled, mazlan2016modelling}. We incorporate stochasticity in our model by adding noise to the cancer equation, and study the differences in long term solution behaviour between our ODE and SDE models.

\section{Carbohydrate Fermentation Model}\label{Gut Model}

The human colon is host to a diverse population of microorganisms, including bacteria, fungi, and archaea, that are collectively referred to as gut microbiota. These microbiota play an essential role in digestion, metabolism, and immune system support. The gut microbiome has become a popular research topic over the past few decades, with increasing evidence highlighting its connection to a wide range of diseases including inflammatory bowel diseases (IBD), irritable bowel syndrome (IBS), obesity, diabetes and cancer \cite{jandhyala2015role, kim2022potential}. The specific composition of the gut microbiome is unique to each individual, depending on factors such as genetics, diet, and antibiotic use \cite{jandhyala2015role}. Fiber fermenting bacteria, including species from the \textit{Bacterioides} and \textit{Bifidobacterium} genera in addition to others \cite{oliphant2019macronutrient}, are responsible for the fermentation of complex carbohydrates into short-chain fatty acids (SCFAs). Acetate, propionate and butyrate are the most abundant SCFAs in the gut. They are essential to human health, playing a significant role in energy metabolism, gut barrier strength and immune response modulation \cite{hou2022gut}.

The human colon is lined with a two-layer mucus membrane: an inner layer that adheres to the epithelial cells and is devoid of bacteria, and an outer, unattached layer which is colonized by microbiota \cite{johansson2011two}. The lumen of the colon is also colonized but encompasses a different microbial composition than that of the mucus layer. Glycans (dietary fiber and mucins) are the main energy sources for carbohydrate fermenting gut microbiota. The majority of dietary fibers originate from plant cell walls, obtained by the consumption of fruits, vegetables and whole grains \cite{williams2017gut}. Dietary fiber escapes digestion in the upper gastrointestinal (GI) tract and enters the lumen of the colon through the small intestine. Conversely, mucins, which are glycoproteins produced by the host's epithelial tissue, are abundant in the mucus layer. Only certain microbial species are capable of degrading fiber and mucins into simpler sugars. These sugars are further broken down by other species into SCFAs, lactate, and gas.

We present a comprehensive model of carbohydrate fermentation in the human colon, which we refer to as the ``gut model". The system we use is a modified version of the models in \cite{moorthy2015spatially, munoz2010mathematical}, consisting of 24 ODEs. We briefly describe the biological mechanisms behind the model, introduce the equations and discuss the effect of prebiotic and probiotic administration on model outputs.

\subsection{Model Description}
 The colon is modelled by a two-compartment bioreactor --- the primary compartment representing the lumen and the second representing the mucus layer. The equations describe exchange between these compartments as well as the host, and the reaction processes involved in the break down of fiber and mucins into SCFAs and byproducts by fiber-fermenting microbiota. Microbiota are categorized into three main functional groups - sugar-degraders (SDs), lactate degraders (LDs) and hydrogen-degrading acetogens (HDAs), which are each present in both the lumen and mucus compartments.

Sugar degraders break down fiber, which is supplied at a constant rate into the lumen, as well as mucins, which are produced in the mucus layer, into sugars. Sugars can diffuse between the lumen and mucus and are further broken down into hydrogen, lactate, and SCFAs by sugar degraders in both compartments. Lactate and hydrogen are also broken down to SCFAs and byproducts (carbon dioxide and water) by lactate degraders and hydrogen degrading acetogens, respectively. Each of these compounds along with the biomass can move between compartments through attachment (lumen to mucus) or sloughing (mucus to lumen). Lactate, SCFAs, water and carbon dioxide in the mucus compartment can be absorbed by the host. Diagrams of the bioreactor setup and main reaction processes are shown in Figure \ref{fig: gut diagrams}.

\begin{figure}
    \centering
    \includegraphics[width=1.0\linewidth]{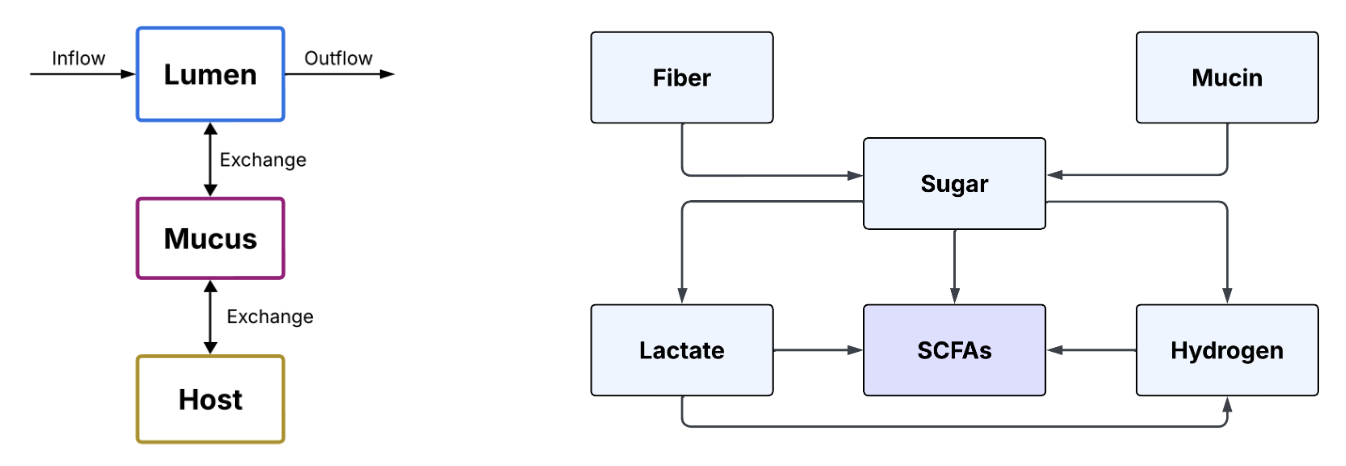}  
    \caption[Diagram of bioreactor setup and fermentation reactions]{Diagrams of the gut bioreactor (left) and reactions involved in the fermentation of fiber and mucins into SCFAs (right). This reaction pathway occurs in both the lumen and mucus compartments of the bioreactor.}
    \label{fig: gut diagrams}
\end{figure}

The dependent variables are summarized in Table \ref{tab: dependent variables} in units of g/L. Subscripts $l$ and $m$ on variables indicate their lumen and mucus concentrations respectively. Model parameters and their default values are given in Tables \ref{tab:yield coeffs}-\ref{tab: other params}. Default values of initial biomass concentrations in units of $g/L$ are given in Table \ref{tab: gut ics}. Initial concentrations of the remaining variables are all zero. Simulations were run with these parameter values unless otherwise stated.

\begin{table}[h]
    \centering
    \caption[Dependent variables in the gut model]{Dependent variables of the gut model where $i\in\{l,m\}$ denotes lumen ($l$) or mucus ($m$).}
    \vspace{0.1cm}
    \begin{tabular}{c|l}
        Symbol & Variable (g/L)\\ \hline
        $I_l$ & Fiber (lumen only)\\
        $I_m$ & Mucin (mucus only)\\
        $S_i$ & Sugar\\
        $L_i$ & Lactate\\
        $H_i$ & Hydrogen\\
        $A_i$ & Acetate\\
        $P_i$ & Propionate\\
        $B_i$ & Butyrate\\
        $C_i$ & Carbon dioxide\\
        $W_i$ & Water\\
        $X_{1,i}$ & SD biomass \\
        $X_{2,i}$ & LD biomass \\
        $X_{3,i}$ & HDA biomass
    \end{tabular}
    \label{tab: dependent variables}
\end{table}

 The dependent variables are grouped into vectors $\mathbf{c_l,c_m}\in\mathbb{R}^{12}$ where
\begin{align*}
    \mathbf{c_l}&=(I_l,S_l,L_l,H_l,A_l,P_l,B_l,C_l,W_l,X_{1,l},X_{2,l},X_{3,l})^T,\\
     \mathbf{c_m}&=(I_m,S_m,L_m,H_m,A_m,P_m,B_m,C_m,W_m,X_{1,m},X_{2,m},X_{3,m})^T.
\end{align*}
We define the reaction function $\rho:\mathbb{R}^{12}\to\mathbb{R}^7$ as
\begin{align*}
    \rho(\mathbf{c_i})=\left(g_{I}(I_i,X_{1,i}),f_{S}(S_i,X_{1,i}),f_{L}(L_i,X_{2,i}),f_{A}(A_i,X_{3,i}),\phi_1(X_{1,i}),\phi_2(X_{2,i}),\phi_3(X_{3,i})\right)^T
\end{align*}
for $i\in\{l,m\}$. The reaction processes are summarized in Table \ref{tab: rxns}. The Peterson matrix $M\in\mathbb{R}^{7\times12}$ of yield coefficients is defined by

\renewcommand{\arraystretch}{1.0}
\begin{align*}
    M=
    \begin{pmatrix}  
        -1  & Y_{SI} & 0      & 0      & 0      & 0      & 0      & 0      & 0      & 0     & 0     & 0 \\
        0   &  -1    & Y_{LS} & Y_{HS} & Y_{AS} & Y_{PS} & Y_{BS} & Y_{CS} & Y_{WS} & Y_{S} & 0     & 0 \\ 
        0   &   0    & -1     & Y_{HL} & Y_{AL} & Y_{PL} & Y_{BL} & Y_{CL} & Y_{WL} & 0     & Y_{L} & 0 \\ 
        0   &   0    & 0      & -1     & Y_{AH} & 0      & 0      & Y_{CH} & Y_{WH} & 0     & 0     & Y_{H} \\ 
        0   &   0    & 0      & 0      & 0      & 0      & 0      & 0      & 0      & -1    & 0     & 0 \\ 
        0   &   0    & 0      & 0      & 0      & 0      & 0      & 0      & 0      & 0     & -1    & 0 \\ 
        0   &   0    & 0      & 0      & 0      & 0      & 0      & 0      & 0      & 0     & 0     & -1 
    \end{pmatrix},
\end{align*}
where rows correspond to reaction processes and columns to dependent variables. The exchange matrices $E_i\in\mathbb{R}^{12\times12}$ are diagonal 
with entries $E_i(j,j)=\gamma_{i,c_j}$ for $j=1,2,...,12$ and $i=\{1,2,3,4\}$ (corresponding to sloughing, attachment, diffusion and absorption, respectively). Exchange rates $\gamma_{i,c_j}$ not listed in Table {\ref{tab: exchange rates}} are taken to be zero, in which case the dependent variable $c_j$ does not undergo the exchange process indexed by $i$.

Finally, $\mathbf{c}_\text{inf}, \mathbf{\Lambda}\in\mathbb{R}^{12}$ are vectors of inflow rates and production rates respectively, each of which has zeros in all entries except the first, since only the equation for fiber ($I_l$) has an inflow term and only mucins ($I_m$) have a source:
\begin{align*}    
\mathbf{c}_\text{inf}=I_\text{inf}\,\mathbf{e_1},
\hspace{0.6cm}
\mathbf{\Lambda}=
\begin{cases}
    \mathbf{0}  &\text{ if }I_m>I_\text{max},\\[0.5em]
    \Gamma\left(1-\frac{I_m}{I_\text{max}}\right)\mathbf{e_1} &\text{otherwise,}
\end{cases}
\end{align*}
where $\mathbf{0}$ and $\mathbf{e_1}$ denote the zero vector and the first basis vector of of $\mathbb{R}^{12}$, respectively. The equations in vector notation are then
\begin{align}
        \frac{d\mathbf{c_l}}{dt} 
    &= \underbrace{M^T\rho(\mathbf{c_l})}_{\text{reactions}}
    +\underbrace{\frac{V_m}{V_l}E_1\mathbf{c_m}}_{\text{sloughing}}-\underbrace{E_2\mathbf{c_l}}_{\text{attachment}}
    -\underbrace{\frac{1}{V_l}E_3(\mathbf{c_l}
    -\mathbf{c_m})}_{\text{diffusion}}+q(\underbrace{\mathbf{c}_\text{inf}}_{\text{inflow}}-\underbrace{\mathbf{c_l}}_{\text{outflow}}),\label{eq: gut l}\\
    \frac{d\mathbf{c_m}}{dt}
    &=\underbrace{M^T\rho(\mathbf{c_m})}_{\text{reactions}} 
    -\underbrace{E_1\mathbf{c_m}}_{\text{sloughing}}
    +\underbrace{\frac{V_l}{V_m}E_2\mathbf{c_l}}_{\text{attachment}}
    +\underbrace{\frac{1}{V_m}E_3(\mathbf{c_l}-\mathbf{c_m})}_{\text{diffusion}} 
    -\underbrace{E_4\mathbf{c_m}}_{\text{absorption}}+\underbrace{\mathbf\Lambda}_{\text{source}} \label{eq: gut m}.
\end{align}

\begin{table}[h]
    \centering
    \caption[Yield coefficients]{Yield coefficients taken from \cite{moorthy2015spatially,munoz2010mathematical}.}
    \vspace{0.1cm}
    \begin{tabular}{c|l|c}
        Symbol & Yield Coefficient & Value \\
        \hline
        $Y_{SI}$ & Sugar from carbohydrate & 1.0 $g_{su}/g_{carb}$\\
        $Y_{LS}$ & Lactate from sugar & 0.09010 $g_{la}/g_{su}$\\
        $Y_{HS}$ & Hydrogen from sugar & 0.00606 $g_{H_2}/g_{su}$\\
        $Y_{AS}$ & Acetate from sugar & 0.18916 $g_{ac}/g_{su}$\\
        $Y_{PS}$ & Propionate from sugar & 0.09877 $g_{pr}/g_{su}$\\
        $Y_{BS}$ & Butyrate from sugar & 0.13215 $g_{bu}/g_{su}$\\
        $Y_{CS}$ & Carbon dioxide from sugar & 0.13333 $g_{CO_2}/g_{su}$\\
        $Y_{WS}$ & Water from sugar & 0.12364 $g_{H_2O}/g_{su}$\\
        $Y_{HL}$ & Hydrogen from lactate & 0.00444 $g_{H_2}/g_{la}$\\
        $Y_{AL}$ & Acetate from lactate & 0.06667 $g_{ac}/g_{la}$\\
        $Y_{PL}$ & Propionate from lactate & 0.16444 $g_{pr}/g_{la}$\\
        $Y_{BL}$ & Butyrate from lactate & 0.09778 $g_{bu}/g_{la}$\\
        $Y_{CL}$ & Carbon dioxide from lactate & 0.14667 $g_{CO_2}/g_{la}$\\
        $Y_{WL}$ & Water from lactate & 0.2 $g_{H_2O}/g_{la}$\\
        $Y_{AH}$ & Acetate from hydrogen & 2.14286 $g_{ac}/g_{H_2}$\\
        $Y_{CH}$ & Carbon dioxide from hydrogen & -11.0 $g_{CO_2}/g_{H_2}$\\
        $Y_{WH}$ & Water from hydrogen & 6.42857 $g_{H_2O}/g_{H_2}$\\
        $Y_S$    & SD bacteria from sugar & 0.3424 $g_{X_1}/g_{su}$\\
        $Y_L$    & LD bacteria from lactate & 0.37667 $g_{X_2}/g_{la}$\\
        $Y_H$    & HDA bacteria from hydrogen & 4.035714 $g_{X_3}/g_{H_2}$\\
    \end{tabular}
    \label{tab:yield coeffs}
\end{table}

\begin{table}[]
    \centering
    \caption[Exchange rates]{Exchange rates taken from \cite{moorthy2015spatially,munoz2010mathematical}. Exchange rates not listed here are taken to be zero.}
    \vspace{0.1cm}
    \begin{tabular}{c|c|c}
    Symbol & Parameter & Value (d$^{-1}$)\\ \hline
    $\gamma_{1,I}$ & Mucin sloughing rate & 0.1 \\
    $\gamma_{1,X_1}$ & SD sloughing rate & 0.4\\
    $\gamma_{1,X_2}$ & LD sloughing rate & 0.4\\
    $\gamma_{1,X_3}$ & HDA sloughing rate & 0.4\\
    $\gamma_{2,L}$ & Lactate attachment rate & 0.88 \\
    $\gamma_{2,A}$ & Acetate attachment rate & 1.32\\
    $\gamma_{2,P}$ & Propionate attachment rate & 1.07 \\
    $\gamma_{2,B}$ & Butyrate attachment rate & 0.9\\
    $\gamma_{2,C}$ & Carbon dioxide attachment rate & 1.0\\
    $\gamma_{2,W}$ & Water attachment rate & 1.0\\
    $\gamma_{2,X_1}$ & SD attachment rate & 0.1\\
    $\gamma_{2,X_2}$ & LD attachment rate & 0.1\\
    $\gamma_{2,X_3}$ & HDA attachment rate & 0.1\\    
    $\gamma_{3,S}$ & Sugar diffusion rate & 3.9\\
    $\gamma_{4,L}$ & Lactate absorption rate & 12.6\\
    $\gamma_{4,A}$ & Acetate absorption rate & 18.9\\
    $\gamma_{4,P}$ & Propionate absorption rate & 15.32\\
    $\gamma_{4,B}$ & Butyrate absorption rate & 12.88\\
    $\gamma_{4,C}$ & Carbon dioxide absorption rate & 14.0\\
    $\gamma_{4,W}$ & Water absorption rate & 1.6\\
    \end{tabular}
    \label{tab: exchange rates}
\end{table}

\begin{table}[]
    \centering
    \caption[Reaction parameters]{Reaction parameters taken from \cite{moorthy2015spatially}.}
    \vspace{0.1cm}
    \begin{tabular}{c|c|c}
        Symbol & Parameter & Value \\ \hline
        $\kappa_I$ & Hydrolysis rate & 10.62 $g_I/(g_{X_{\text{su}}}\cdot d)$\\
        $\kappa_S$ & Sugar consumption rate & 12.63 $g_\text{su}/(g_{X_{\text{su}}}\cdot d)$\\
        $\kappa_L$ & Lactate consumption rate & 82.11 $g_\text{la}/(g_{X_{\text{la}}}\cdot d)$\\
        $\kappa_S$ & Hydrogen consumption rate & 1.93 $g_{H_2}/(g_{X_{H_2}}\cdot d)$\\
        $K_I$ & Hydrolysis concentration ratio & 0.27 $g_I/g_{X_{\text{su}}}$\\
        $K_S$ & Sugar half saturation constant & 0.47 $g_\text{su}/L$\\
        $K_L$ & Lactate half saturation constant & 0.60 $g_\text{la}/L$\\
        $K_H$ & Hydrogen half saturation constant & 0.0034 $g_{H_2}/L$\\
        $\kappa_1$ & SD biomass decay & 0.01 $d^{-1}$\\
        $\kappa_2$ & LD biomass decay & 0.01 $d^{-1}$\\
        $\kappa_3$ & HDA biomass decay & 0.01 $d^{-1}$\\
    \end{tabular}
    \label{tab: rxn params}
\end{table}

\begin{table}[]
    \centering
    \caption[Physical model parameters]{Physical model parameters taken from \cite{jegatheesan2024mathematical}.}
    \vspace{0.1cm}
    \begin{tabular}{c|c|c}
        Symbol & Parameter & Value \\ \hline
        $V_l$ & Volume of lumen & 0.9 L\\
        $V_m$ & Volume of mucus & 0.1 L\\
        $q$ & Average flow rate & 1.01 $d^{-1}$\\
        $I_\text{inf}$ & Inflow fiber concentration & 20 $g_I$/L\\
        $I_\text{max}$ & Maximum mucin concentration & 500 $g_I$/L\\
        $\Gamma$ & Mucin production rate & 50 $g_I$/(L$\cdot$d)\\
    \end{tabular}
    \label{tab: other params}
\end{table}

\begin{table}[]
    \centering
    \caption[Default initial biomass concentrations]{Default initial biomass concentrations in units of g/L.}
    \vspace{0.2cm}
    \begin{tabular}{c|c}
        Initial condition & Value (g/L) \\ \hline
        $X_{1,l}(0)$ & 5.5\\
        $X_{1,m}(0)$ & 5.5\\
        $X_{2,l}(0)$ & 2.5\\
        $X_{2,m}(0)$ & 2.5\\
        $X_{3,l}(0)$ & 0.1\\
        $X_{3,m}(0)$ & 0.1
    \end{tabular}
    \label{tab: gut ics}
\end{table}

\renewcommand{\arraystretch}{1.5}
\begin{table}
    \centering
    \caption[Fermentation reaction functions]{Functions describing reaction processes (units of g/(L$\cdot$d)).}
    \vspace{0.1cm}
    \begin{tabular}{c|l|l}
    Index & Rate Description & Function \\ \hline
    1 & Hydrolysis & $g_{I}(I,X_1) = \frac{\kappa_{I}I}{K_{I}X_1 + I}X_1$\\
    2 & Glucose utilization & $f_{S}(S,X_1) = \frac{\kappa_{S}S}{K_{S} + S}X_1$\\    
    3 & Lactate utilization & $f_L(L,X_2) = \frac{\kappa_{L}L}{K_L+L}X_2$ \\
    4 & Acetogenesis & $f_A(A,X_3) = \frac{\kappa_{A}H}{K_A+H}X_3$\\
    5 & $X_1$ Decay & $\phi_1(X_1)=\kappa_1X_1$\\
    6 & $X_2$ Decay & $\phi_2(X_2)=\kappa_2X_2$\\
    7 & $X_3$ Decay & $\phi_3(X_3)=\kappa_3X_3$\\
\end{tabular}
    \label{tab: rxns}
\end{table}

\subsection{Simulation Results}\label{sect: gut results}
Solution components representing SCFA and biomass concentrations in the lumen and mucus are plotted over time in Figure \ref{fig: gut time series1}. At steady state, acetate is the most abundant SCFA in both compartments, which is consistent with relative SCFA levels typically found in healthy patients \cite{cummings1987short, fernandes2014adiposity}. All bacterial functional groups survive, with sugar degraders being the most abundant. The difference in biomass concentrations between the lumen and mucus layers can be attributed to the relative volumes of the two compartments. Since $V_l\gg V_m$, the attachment terms appearing in Equations {(\ref{eq: gut l})} and {(\ref{eq: gut m})} are generally greater in magnitude than the sloughing terms, resulting in more biomass moving to the mucus compartment over time. Figure \ref{fig: gut only vary Iinf} shows the effect of the rate of fiber inflow ($I_\text{inf}$) on steady state values of total biomass and SCFA concentrations. We see that higher values of $I_\text{inf}$ increases each of these concentrations, indicating that in our model, fiber is crucial for gut bacteria and SCFA production, as is known from biological literature \cite{barber2020health}.

\begin{figure}
    \centering
    \includegraphics[width=1.0\linewidth]{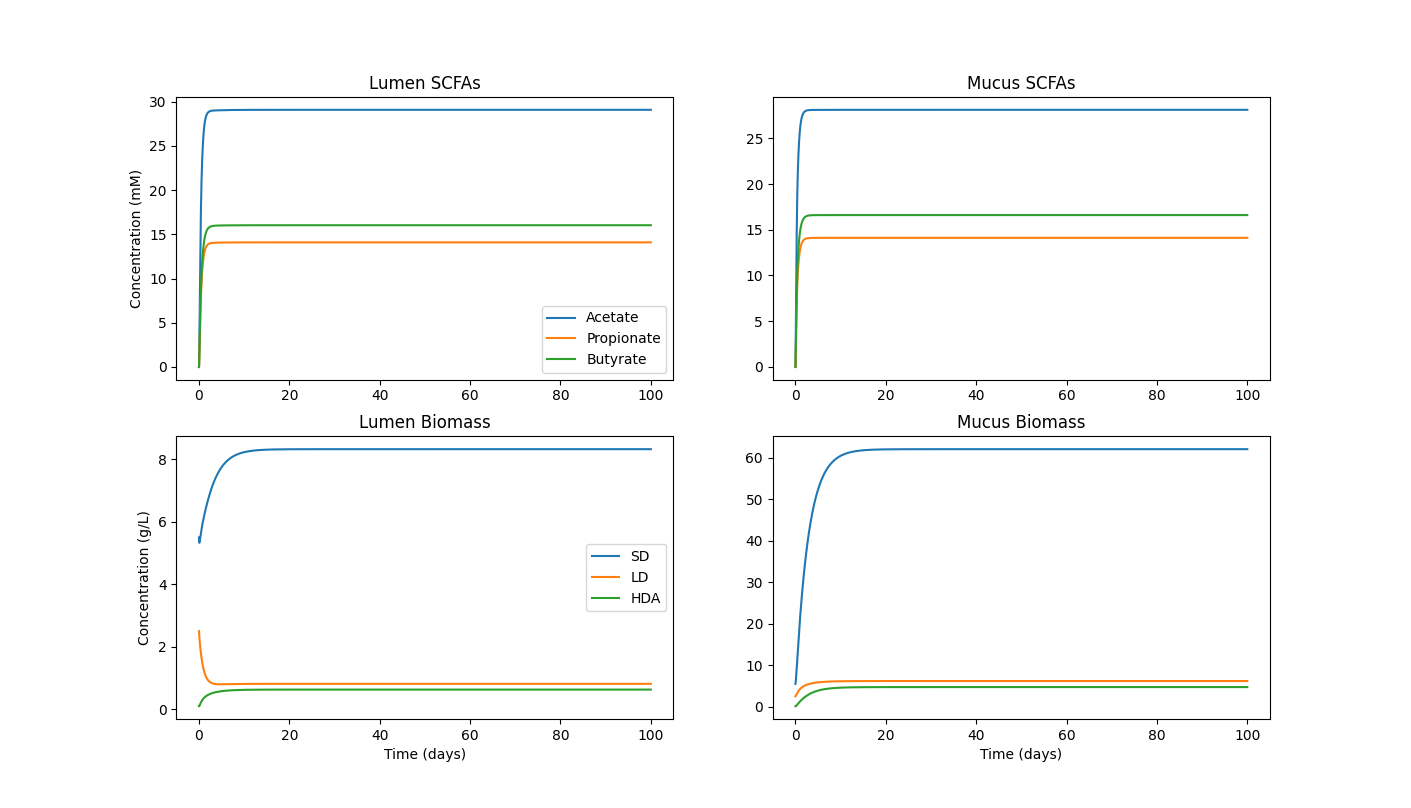}
    \caption[Time series of SCFA and biomass concentrations]{Solution components to the gut model representing SCFA and biomass concentrations over time until steady state in the lumen and mucus compartments. SD, LD and HDA indicate sugar degraders, lactate degraders and hydrogen degrading acetogens, respectively.}
    \label{fig: gut time series1}
\end{figure}

\begin{figure}
    \centering
    \includegraphics[width=1.0\linewidth]{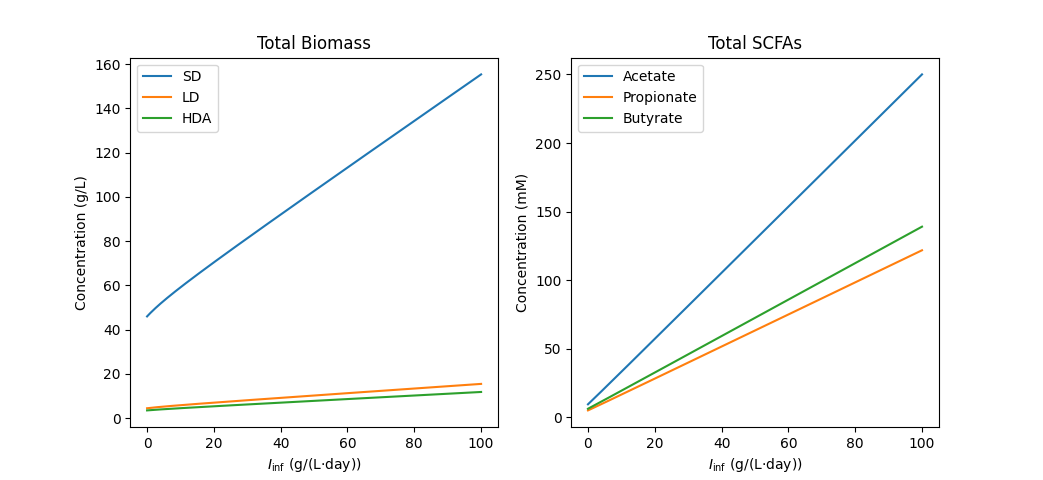}
    \caption[Effect of fiber on SCFAs and biomass]{Effect of varying fiber inflow ($I_\text{inf}$) between 0 and 100 g/L on steady state values of total (lumen + mucus) biomass and SCFA concentrations.}
    \label{fig: gut only vary Iinf}
\end{figure}

We then looked at the effect of initial biomass values on SCFA production. Figure \ref{fig: biomass bifurcation} suggests that the system with the given parameter values is bistable, since as the initial condition of total sugar degraders ($X_{1,l}(0)+X_{1,m}(0)$) is varied, it moves from the basin attraction of a washout equilibrium to that of an equilibrium allowing for bacterial persistence. Without intervention (solid lines), if the initial concentration of sugar degraders exceeds approximately 0.020 g/L, all three functional groups will persist. Otherwise, they will die out. We also find that SCFAs directly rely on biomass concentrations, with their bifurcation between extinction and persistence occurring at the same point as that for the biomass.

The dotted lines show model solutions when a sugar-degrading probiotic is added to the system at a constant dose of 0.035 g/(L$\cdot$day). Probiotic administration expands the basin of attraction of the biomass-existing equilibrium, allowing for persistence of biomass at lower initial concentrations of sugar degraders. However, if initial conditions of sugar degraders exceed approximately 0.020 g/L, then probiotics have no observable effect on SCFA production. This suggests that probiotics may be beneficial for patients with a compromised gut microbiome lacking sugar degraders, but have no significant effect for patients with sufficient sugar degrading microbiota.

\begin{figure}
    \centering
    \includegraphics[width=1.0\linewidth]{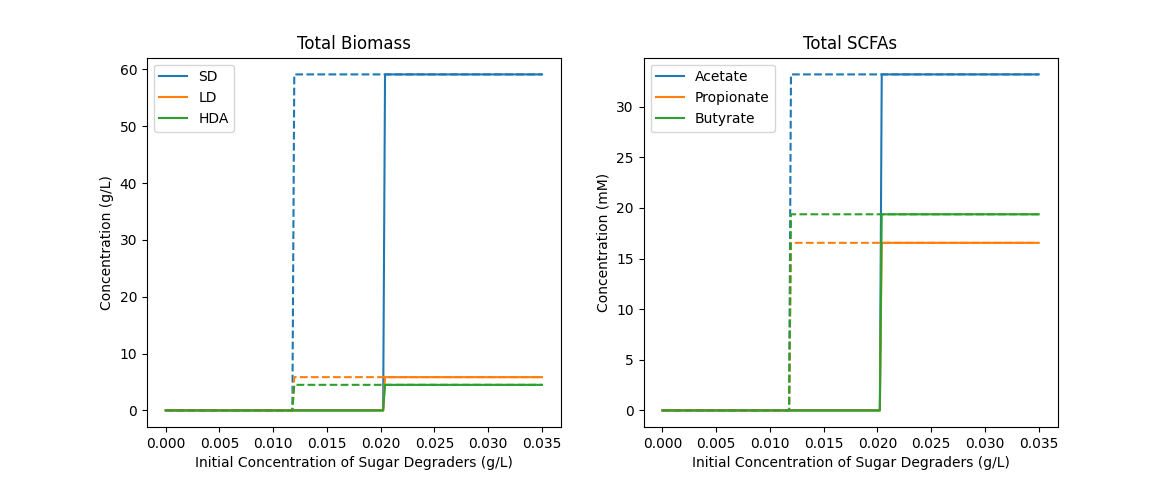}
    \caption[Effect of initial microbial composition on steady state biomass concentrations with and without probiotics]{Effect of varying the initial total concentration of sugar degraders ($X_{1,l}(0)+X_{1,m}(0)$) up to 0.035 g/L on steady states of total biomass (left) and SCFAs (right). Solid lines show model solutions without probiotics whereas dotted lines show the effect of adding probiotics at a constant rate of 0.035 g/(L$\cdot$day).}
    \label{fig: biomass bifurcation}
\end{figure}

\section{Cancer Model Development}

We now introduce our model describing cancer growth in the presence of immune cells and butyrate produced from gut microbiota. The cancer model is constructed so that it can be studied in isolation, or weakly coupled to the gut model introduced in Section \ref{Gut Model}. We begin by outlining the biological assumptions that went into building our model, describe the model equations, prove the existence and uniqueness of solutions, and discuss the analysis of steady states.

\subsection{Model Assumptions}\label{sect: assumptions}
The model is formulated for the microenvironment of a cancerous colon tumour, located along the mucus layer of the colon. The tumour is exposed to various materials in the gut, most notably products from carbohydrate fermentation such as butyrate and other SCFAs. 

We model the dynamics of healthy colonocytes or `normal' cells as well as colon cancer cells. Three key types of immune cells are incorporated in our model, namely cytotoxic T-cells, regulatory T-cells and dendritic cells. Additionally, butyrate produced from carbohydrate fermentation is included. We assume that the concentrations of model constituents are spatially homogeneous, meaning they change only in time and can be modelled by a system of ODEs. The assumptions that drive the model equations are presented below:
\begin{enumerate}
\renewcommand{\labelenumi}{\textbf{A\arabic{enumi}}}
\renewcommand{\theenumi}{A\arabic{enumi}}
    \item \textbf{(Cell growth)} 
    Several functions for modelling cancer growth have been proposed, each with their own advantages in different scenarios \cite{eftimie2011interactions, tabassum2019mathematical}. We choose to model both cancer and normal cell growth logistically, as done in many other cancer-immune models \cite{alharbi2020new, de2001mathematical, raeisi2024mathematical, mohammad2022investigating, kirshtein2020data, de2005validated, dritschel2018mathematical}, as this balances simplicity while still capturing key non-linear dynamics. \label{cell growth}
    \item \textbf{(Uptake)} Butyrate is the primary energy source for colonocytes \cite{donohoe2011microbiome, comalada2006effects}. Butyrate is also absorbed by T cells and Tregs, promoting their growth \cite{luu2021microbial, zhang2025gut, kazmierczak2022sodium, furusawa2013commensal}. \label{B stim}
    \item \textbf{(Competition)} Normal colonocytes compete with cancer for nutrients and space. Noncancerous cells undergo contact inhibition, causing a cease in proliferation and growth when they contact each other to avoid overcrowding. This characteristic is lost in cancer cells \cite{pavel2018contact}, meaning they are not significantly affected by the presence of normal cells.     
    \label{compete}
    \item \textbf{(Inhibition)} In contrast to its positive effect on normal host cells, butyrate inhibits cancer cell growth \cite{li2018butyrate}.\label{B inhib}
    \item \textbf{(Immune activity)} Cytotoxic T cells attack cancer cells by releasing molecules that induce apoptosis \cite{ahmed2023role}. However, excessive T-cell accumulation causes inflammation so they are monitored and down regulated by Tregs \cite{ohue2019regulatory}. Additionally, cancer cells can hinder the antitumour properties of dendritic cells by disrupting their transport and preventing maturation \cite{lin2025new}.\label{inhibitory effects}
    \item \textbf{(Recruitment)} The presence of cancer cells stimulates the recruitment of dendritic cells and Tregs to the tumour microenvironment (TME). Dendritic cells then process and present tumour antigens to the T-cells, promoting their recruitment to the TME \cite{wang2024complex}.\label{recruit}
    \item \textbf{(Death)} Colonocytes, T-cells, Tregs, and dendritic cells die from natural causes, modelled by linear decay as in \cite{de2001mathematical, raeisi2024mathematical, de2005validated}. The natural death rate of cancer cells is assumed to be negligible in comparison.\label{death}
    \item \textbf{(Diffusion)} Butyrate in the mucus layer that is produced during carbohydrate fermentation diffuses into the TME \cite{kazmierczak2022sodium}\label{Binf}.
\end{enumerate}

\subsection{Model Equations}
We now formulate our model based on the above biological assumptions. The state variables are all functions of time $t\geq0$, measured in days, and are summarized in Table \ref{tab: variables}.\\
\begin{table}[h]
    \centering
    \caption[Dependent variables in the cancer model]{Summary of dependent variables in the cancer model.}\label{tab: variables}
    \vspace{0.2cm}
    \begin{tabular}{c|l|c}
        Symbol & Description & Units\\ \hline
        $N$ & Normal colonocytes & cells/mm$^3$ \\
        $C$ & Cancer cells & cells/mm$^3$ \\
        $T$ & Cytotoxic T cells & cells/mm$^3$ \\
        $R$ & Regulatory T cells (Tregs) & cells/mm$^3$ \\
        $D$ & Dendritic cells & cells/mm$^3$ \\
        $B$ & Butyrate & mM
    \end{tabular}
\end{table}

\noindent\textbf{Normal Cells:}\\
As per \ref{cell growth}, normal cells grow logistically with constant growth rate $r_N$ and carrying capacity $k_N$. By \ref{B stim}, growth is enhanced by the presence of butyrate, modelled by a monotonically increasing sigmoidal function of $B$ with half saturation constant $\kappa_B$:
\begin{align*}
    g(B)=\frac{B}{B+\kappa_B}.
\end{align*}
Normal cells decay at a rate proportional to the number of cancer cells as in \cite{alharbi2020new, de2001mathematical, raeisi2024mathematical}, due to competition for space and other resources (\ref{compete}). The cells die from natural causes at rate $\alpha_N$ (\ref{death}). The equation for $N$ is then 
\begin{align*}
    N' &= \underbrace{r_NN\left(1-\frac{N}{k_N}\right)}_{\ref{cell growth}}\,\underbrace{g(B)}_{\ref{B stim}}-\underbrace{\delta_NNC}_{\ref{compete}}-\underbrace{\alpha_NN}_{\ref{death}}.
\end{align*}

\noindent\textbf{Cancer Cells:}\\
Cancer cells also grow logistically, by \ref{cell growth}, with growth rate $r_C$ and carrying capacity $k_C$. As opposed to normal cells, cancer cell growth is inhibited by the presence of butyrate, as per assumption \ref{B inhib}. This is modelled by a $B$-dependent monotonically decreasing sigmoidal function,
\begin{align*}
    I(B)=\frac{\kappa_B}{B+\kappa_B}.
\end{align*}
By assumption \ref{inhibitory effects}, cancer cells die proportional to T cell density at rate $\delta_C$, as in \cite{alharbi2020new, de2001mathematical, dritschel2018mathematical, raeisi2024mathematical}, giving the equation for $C$ as
\begin{align*}
    C' &= \underbrace{r_CC\left(1-\frac{C}{k_C}\right)}_{\ref{cell growth}}\,\underbrace{I(B)}_{\ref{B inhib}}-\underbrace{\delta_C\,CT}_{\ref{inhibitory effects}}.
\end{align*}

\noindent\textbf{T Cells:}\\
T cell proliferation is modelled linearly in $T$ with rate $r_T$ and is enhanced by butyrate (\ref{B stim}). T cell recruitment is directly influenced by the number of dendritic cells in the TME (\ref{recruit}) and, similar to \cite{raeisi2024mathematical}, is modelled as an inflow proportional to $D$ at rate $\varphi_T$. By \ref{inhibitory effects}, T cells are down-regulated proportional to Treg density at rate $\delta_T$. T cells die from natural causes at rate $\alpha_T$, by \ref{death}, giving their equation as
\begin{align*}
    T' &= \underbrace{{r_T\,g(B)}\,T}_{\ref{B stim}}+\underbrace{\varphi_TD}_{\ref{recruit}}-\underbrace{\delta_T\,TR}_{\ref{inhibitory effects}}-\underbrace{\alpha_T\,T}_{\ref{death}}.
\end{align*}

\noindent\textbf{Regulatory T Cells:}\\
By assumption \ref{B stim}, Treg proliferation is also enhanced by butyrate and is modelled analogously to that of T cells. By \ref{recruit}, Tregs are recruited to the TME at rate $\varphi_R$, proportional to the cancer cell density and die at rate $\alpha_R$ (\ref{death}):
\begin{align*}
    R' &= \underbrace{r_R\,g(B)\,R}_{\ref{B stim}}+\underbrace{\varphi_R\,C}_{\ref{recruit}}-\underbrace{\alpha_RR}_{\ref{death}}.
\end{align*}

\noindent\textbf{Dendritic Cells:}\\
Like Tregs, dendritic cells are recruited to the TME by cancer cells (\ref{recruit}) at rate $\varphi_D$ but decay proportionally to cancer cell density (\ref{inhibitory effects}) at rate $\delta_D$ and die at rate $\alpha_D$ (\ref{death}). This is similar to the DC equation in \cite{raeisi2024mathematical}:
\begin{align*}
    D' &= \underbrace{\varphi_D\,C}_{\ref{recruit}}-\underbrace{\delta_D\,DC}_{\ref{inhibitory effects}}-\underbrace{\alpha_DD}_{\ref{death}}.
\end{align*}

\noindent\textbf{Butyrate:}\\
The mass transfer of butyrate occurs at rate $q_B$ and diffuses into the TME depending on the concentration gradient between the TME and its surroundings (\ref{Binf}). When coupled to the gut model, $B_\text{inf}$ is the solution component representing butyrate concentration in the mucus layer produced during fermentation ($B_\text{inf}(t)=B_m(t)$). In the uncoupled case, $B_\text{inf}$ is taken to be constant. Normal, T, and regulatory T cells uptake butyrate from the TME at rates $u_N$, $u_T$ and $u_R$ respectively (\ref{B stim}):
\begin{align*}
    B' &= \underbrace{q_B(B_{\text{inf}}-B)}_{\ref{Binf}}-\underbrace{g(B)\,(u_N\,N+u_T\,T+u_R\,R)}_{\ref{B stim}}.
\end{align*}

The full cancer model is then given by 
     \begin{align}
        N' &= \underbrace{r_NN\left(1-\frac{N}{k_N}\right)}_{\ref{cell growth}: \text{ Cell growth}}\,\underbrace{\frac{B}{B+\kappa_B}}_{\ref{B stim}\text{: Uptake}}-\underbrace{\delta_NNC}_{\ref{compete}: \text{ Competition} }-\underbrace{\alpha_NN}_{\ref{death}: \text{ Death}},\label{eq: N}\\
        C' &= \underbrace{r_CC\left(1-\frac{C}{k_C}\right)}_{\ref{cell growth}: \text{ Cell growth}}\,\underbrace{\frac{\kappa_B}{B+\kappa_B}}_{\ref{B inhib}:\text{ Inhibition}}-\underbrace{\delta_C\,CT}_{\ref{inhibitory effects}: \text{ Immune activity}},\label{eq: C}\\
        T' &= \underbrace{{r_T\,\frac{B}{B+\kappa_B}}\,T}_{\ref{B stim}:\text{ Uptake}}+\underbrace{\varphi_TD}_{\ref{recruit}:\text{ Recruitment}}-\underbrace{\delta_T\,TR}_{\ref{inhibitory effects}:\text{ Immune activity}}-\underbrace{\alpha_T\,T}_{\ref{death}:\text{ Death}},\\
        R' &= \underbrace{r_R\,\frac{B}{B+\kappa_B}\,R}_{\ref{B stim}:\text{ Uptake}}+\underbrace{\varphi_R\,C}_{\ref{recruit}: \text{ Recruitment}}-\underbrace{\alpha_RR}_{\ref{death}:\text{ Death}},\\
        D' &= \underbrace{\varphi_D\,C}_{\ref{recruit}:\text{ Recruitment}}-\underbrace{\delta_D\,DC}_{\ref{inhibitory effects}:\text{ Immune activity}}-\underbrace{\alpha_DD}_{\ref{death}:\text{ Death}},\\
         B' &= \underbrace{q_B(B_{\text{inf}}-B)}_{\ref{Binf}:\text{ Diffusion}}-\underbrace{\frac{B}{B+\kappa_B}\,(u_N\,N+u_T\,T+u_R\,R)}_{\ref{B stim}:\text{ Uptake}}\label{eq: B}.
    \end{align} 
subject to initial data
\begin{align}
    (N(0),C(0),T(0),R(0),D(0),B(0))=(N_0,C_0,T_0,R_0,D_0,B_0)\in\mathbb{R}_{+,0}^6\label{initial data}.
\end{align}

    
\subsection{Existence and Uniqueness}

\begin{theorem} \label{thm: det existence}
For any $T>0$, the initial value problem of (\ref{eq: N})--(\ref{eq: B}) with initial data (\ref{initial data}) has a unique solution up to time $T$ that remains in the non-negative cone $\mathbb{R}_{+,0}^6$ and depends continuously on initial conditions and model parameters. Furthermore, the solution is bounded above in $\mathbb{R}_{+,0}^6$ with constant bounds on $N(t)$, $C(t)$, $D(t)$ and $B(t)$ given by
    \begin{align*}
         &N(t)\leq\max\{k_N,N_0\}, \hspace{0.5cm} C(t)\leq\max\{k_C,C_0\}:=M_C,\\
        &D(t)\leq\max\{\varphi_D/\delta_D,D_0\}, \hspace{0.5cm} B(t)\leq\max\{B_{\text{inf}},B_0\},
    \end{align*}
    and time-dependent bounds on $T(t)$ and $R(t)$ given by
    \begin{align*}
        T(t)\leq 
    \begin{cases}
        \max\bigg\{\frac{\varphi_TM_D}{\alpha_T-r_T},T_0\bigg\}\hspace{0.2cm}\text{if}\hspace{0.2cm}r_T<\alpha_T,\\        \varphi_TM_D\,t+T_0\hspace{0.2cm}\text{if}\hspace{0.2cm}r_T=\alpha_T,\\
        \left(T_0+\frac{\varphi_T M_D}{r_T-\alpha_T}\right)e^{(r_T-\alpha_T)t}\hspace{0.2cm}\text{if}\hspace{0.2cm}r_T>\alpha_T,
    \end{cases}
    \end{align*}
    and
    \begin{align*}
    R(t)\leq 
    \begin{cases}
        \max\bigg\{\frac{\varphi_rk_C}{\alpha_R-r_R},R_0\bigg\}\hspace{0.2cm}\text{if}\hspace{0.2cm}r_R<\alpha_R,\\        \varphi_Rk_C\,t+R_0\hspace{0.2cm}\text{if}\hspace{0.2cm}r_R=\alpha_R,\\
        \left(R_0+\frac{\varphi_R k_C}{r_R-\alpha_R}\right)e^{(r_R-\alpha_R)t}\hspace{0.2cm}\text{if}\hspace{0.2cm}r_R>\alpha_R,
    \end{cases}
    \end{align*}
    where $M_D=\varphi_D\,M_C$ is a positive constant.
\end{theorem}

\begin{proof}
    Let $x(t)=(N(t),C(t),T(t),R(t),D(t),B(t))$. Given $x(0)\in\mathbb{R}_{+,0}^6$, $x(t)$ remains in $\mathbb{R}_{+,0}^6$ for all $t>0$. This follows from the non-negativity of $x_i'(t)$ along $x_i(t)=0$ for each $i\in\{1,2,..,6\}$, showing positive invariance of $\mathbb{R}_{+,0}^6$ under the flow of the system. The system (\ref{eq: N})--(\ref{eq: B}) satisfies a local Lipschitz condition in the non-negative cone $\mathbb{R}_{+,0}^6$ and by the Picard-Lindelof theorem, this proves the local existence and uniqueness of a solution that depends continuously on initial conditions and parameters. 

    To prove the bounds, we show arguments for $N$ and $T$ with the others following similarly. For $N$, first assume $N_0\in(0,k_N)$ and suppose $N(\tilde t)>k_N$ for some $\tilde t>0$. Then there is some $0<t^*<\tilde t$ such that $N(t^*)=k_N$ and $N'(t^*)\geq0$. But plugging in $N(t^*)=k_N$ into (\ref{eq: N}) gives $N'(t^*)=-\delta_Nk_NC(t^*)-\alpha_Nk_N<0$, contradicting the supposition. If $N_0>k_N$ then for any $t>0$ satisfying $N(t)>k_N$, we have $N'(t)<0$. Therefore, $N(t)\leq\max\{k_N,N_0\}$. The bounds on $C$, $D$ and $B$ are shown similarly.

    For $T$, the $r_T<\alpha_T$ case follows from similar arguments for $N$. The remaining cases rely on the ODE comparison principle. If $r_T=\alpha_T$ then $T'(t)\leq\varphi_TM_D$ so $T(t)\leq\varphi_TM_D\,t+T_0$. If $r_T>\alpha_T$ then $ T'(t)\leq(r_T-\alpha_T)\,T(t)+\varphi_TM_D$. Let $u(t)$ solve
\begin{align*}
    u'(t)=(r_T-\alpha_T)u(t)+\varphi_TM_D,\hspace{0.5cm}u(0)=T_0
\end{align*}
and set $v(t)=u(t)+\frac{\varphi_TM_D}{r_T-\alpha_T}$. Then $v(t)>u(t)$ and 
\begin{align*}
    v'(t)=u'(t)&=(r_T-\alpha_T)u(t)+\varphi_TM_D\\
    &=(r_T-\alpha_T)\left(v(t)-\frac{\varphi_TM_D}{r_T-\alpha_T}\right)+\varphi_TM_D\\
    &=(r_T-\alpha_T)v(t)
\end{align*}
so 
\begin{align*}
    v(t)=v(0)e^{(r_T-\alpha_T)t}=\left(u(0)+\frac{\varphi_TM_D}{r_T-\alpha_T}\right)e^{(r_T-\alpha_T)t}
\end{align*}
and $T(t)\leq u(t)<v(t)$ so we have
\begin{align*}
    T(t)\leq\left(T_0+\frac{\varphi_TM_D}{r_T-\alpha_T}\right)e^{(r_T-\alpha_T)t}.
\end{align*}
The bound on $R$ is shown similarly. Global well posedness then follows from the local result along with boundedness in $\mathbb{R}^6_{+,0}$.
\end{proof}

\begin{table}
\caption[Default cancer model parameters]{Default parameter values used in numerical simulations. Parameters without references are newly introduced.}\label{tab: parameters}
\vspace{0.2cm}
\begin{tabular}{c|l|l|l|c}
Symbol & Description & Value & Units & Reference \\
\hline
$r_N$ & Growth rate of colonocytes & 0.3 & d$^{-1}$ & \cite{reynolds2014canonical}\\
$r_C$ & Growth rate of cancer cells & 0.05 & d$^{-1}$ & \cite{raeisi2024mathematical, burke2020tumour, balci2026quantifying}  \\
$r_T$ & Proliferation rate of T cells & 0.02 & d$^{-1}$ & \cite{balci2026quantifying}\\
$r_R$ & Proliferation rate of Tregs & 0.02 & d$^{-1}$ & -\\
$k_N$ & Normal cell carrying capacity & $5.0\times10^4$ & cells\,mm$^{-3}$ & -\\
$k_C$ & Cancer cell carrying capacity & $1.0\times10^5$ & cells\,mm$^{-3}$ & \cite{eftimie2011interactions, balci2026quantifying} \\
$\kappa_B$ & Butyrate half saturation constant & 1.0 & mM & \cite{han2018butyrate}\\
$\delta_N$ & Competition coefficient
& $1.0\times10^{-6}$ & (d\,cells)$^{-1}$\,mm$^{3}$ & -\\
$\delta_C$ & Effect of T cells on cancer cells
& $5.0\times10^{-6}$ & (d\,cells)$^{-1}$\,mm$^{3}$ & - \\
$\delta_T$ & Effect of Tregs on T cells 
& $1.0\times10^{-5}$ & (d\,cells)$^{-1}$\,mm$^{3}$& - \\
$\delta_D$ & Cancer cell effect on DCs 
& $1.0\times10^{-6}$ & (d\,cells)$^{-1}$\,mm$^{3}$ & \cite{raeisi2024mathematical}\\
$\alpha_N$ & Death rate of colonocytes & 0.1 & d$^{-1}$ & \\
$\alpha_T$ & Death rate of T cells & 0.03 & d$^{-1}$ & \cite{raeisi2024mathematical, wang2023mathematical}\\
$\alpha_R$ & Death rate of Tregs & 0.05 & d$^{-1}$ & -\\
$\alpha_D$ & Death rate of DCs & 0.03 & d$^{-1}$ & \cite{raeisi2024mathematical}\\
$\phi_T$ & T cell inflow rate & 0.8 & d$^{-1}$ & - \\
$\phi_R$ & Treg inflow rate & $1.0\times10^{-3}$ & d$^{-1}$ & - \\
$\phi_D$ & Dendritic cell inflow rate & $1.0\times10^{-3}$ & d$^{-1}$ & \cite{raeisi2024mathematical}\\
$q_B$ & Butyrate mass transfer coefficient & 0.01 & d$^{-1}$ & -\\
$B_{\text{inf}}$ & Inflow concentration of butyrate & [0-50] & mM & \cite{lecona2008kinetic}\\
$u_N$ & Colonocyte butyrate uptake rate
& $5\times10^{-6}$ & mM\,(d\,cells)$^{-1}$\,mm$^{3}$ & \cite{lecona2008kinetic}\\
$u_T$ & T cell butyrate uptake rate
& $5\times10^{-6}$ & mM\,(d\,cells)$^{-1}$\,mm$^{3}$  & -\\
$u_R$ & Treg butyrate uptake rate
& $5\times10^{-6}$ & mM\,(d\,cells)$^{-1}$\,mm$^{3}$  & -\\
\end{tabular}
\end{table}

\begin{table}[]
    \centering
    \caption[Default initial conditions]{Default initial conditions used in numerical simulations.}
    \vspace{0.2cm}
    \begin{tabular}{c|c|c|c}
        Symbol & Description & Value & Units\\ \hline
        $N_0$ & Initial density of normal cells & $k_N$ & cells/mm$^3$\\
        $C_0$ & Initial density of cancer cells & 10,000 & cells/mm$^3$\\
        $T_0$ & Initial density of T cells & 0 & cells/mm$^3$\\
        $R_0$ & Initial density of Tregs & 0 & cells/mm$^3$\\
        $D_0$ & Initial density of dendritic cells & 0 & cells/mm$^3$\\
        $B_0$ & Initial butyrate concentration & 0 & mM\\
    \end{tabular}
    \label{tab: ics}
\end{table}

\subsection{Sensitivity Analysis}\label{sec: SA}
Sensitivity analysis is a standard method for exploring the parameter space of ODE models. We perform both local and global sensitivity analyses for the cancer model, to investigate the effect of each parameter on the density of cancer cells at one year.

In local sensitivity analysis, sensitivity coefficients describe the effect on model outputs of varying parameters individually by some small amount. We consider the cancer cell density at a fixed time point (one year) as our output and treat the model parameters as inputs. If we let $X_i$, for $i\in\{1,2,...,20\}$, represent the $i^\text{th}$ model parameter with default value $X^*_i$, then we calculate the (normalized) sensitivity coefficient $s_i$ for cancer cell density at one year $C(t^*,X_i)$, where $t^*=365$ days, as
\begin{align*}
    s_i=\frac{X_i^*}{C(t^*,X_i^*)}\cdot\frac{\partial C(t^*,X_i)}{\partial X_i}\bigg|_{X_i=X_i^*}.
\end{align*}
The sensitivity coefficients for parameters at their default values, calculated via forward finite difference approximations with a step size of $(3\times10^{-4})\cdot X_i^*$, are shown in Figure \ref{fig: local SA} for different default values of $B_\text{inf}$. 
\begin{figure}
    \centering
    \includegraphics[width=0.49\linewidth]{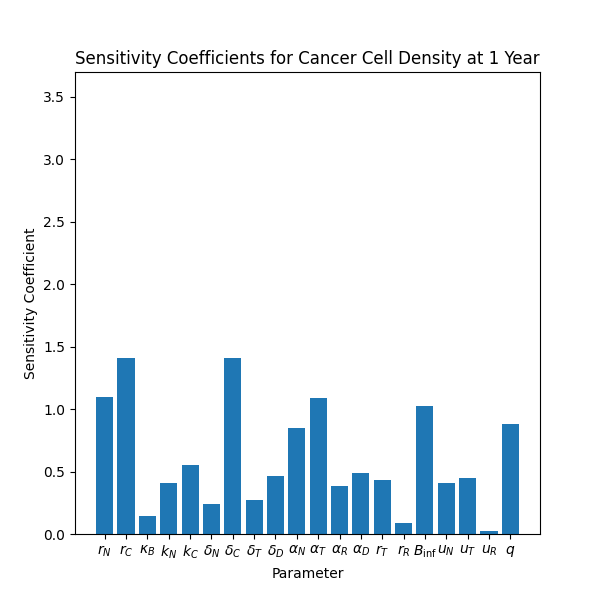}
    \includegraphics[width=0.49\linewidth]{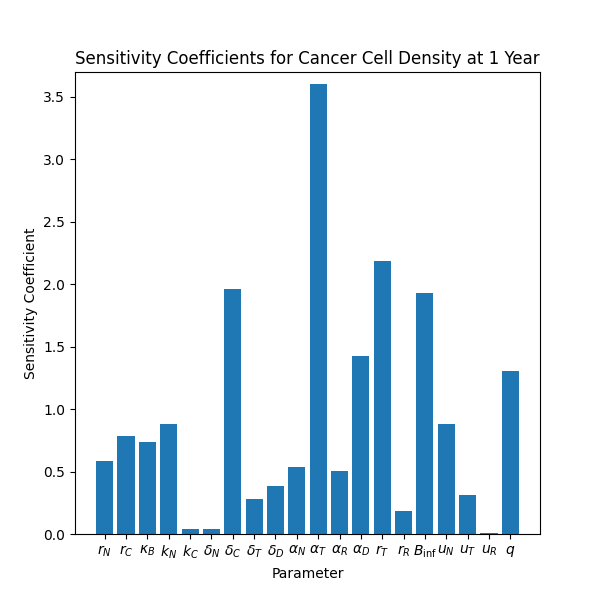}
    \caption[Sensitivity coefficients for model parameters.]{Normalized sensitivity coefficients for model parameters (at their default values) on cancer cell density at one year with $B_\text{inf}=10$ mM (left) and $B_\text{inf}=30$ mM (right).}
    \label{fig: local SA}
\end{figure}

Local sensitivity analysis is only informative for model behaviour near a single parameter set and may not capture interactions between parameters. This motivates our use of global sensitivity analysis, which measures the impact of varying parameters within arbitrarily wide ranges. We use Sobol sensitivity analysis (SSA), which is a variance-based form of global sensitivity analysis. SSA provides sensitivity indices for each model input that measure the effect of changing the input on the variance of the output. First order indices indicate the individual effect of varying a single input, whereas total order indices include interactions with other input parameters. For a given parameter $X_i$, the corresponding first and total order Sobol indices are given by
\begin{align*}
    S_i=\frac{\text{Var}(E[Y|X_i])}{\text{Var(Y)}} \hspace{0.5cm}\text{and}\hspace{0.5cm}S_i^T=1-\frac{\text{Var}(E[Y|X_{\sim i}])}{\text{Var(Y)}},
\end{align*}
respectively, where $Y$ is the density of cancer cells at one year and $X_{\sim i}$ indicates the set of all model parameters except for $X_i$. Model outputs are more sensitive to inputs with higher sensitivity indices. We used the Python package {\tt SALib} with the Saltelli sampling scheme \cite{Herman2017} to conduct SSA, sampling parameters from a range within $\pm15\%$ of their default values. The first and total order Sobol indices are plotted in Figure {\ref{fig: SSA}}.

\begin{figure}
    \centering
    \includegraphics[width=0.49\linewidth]{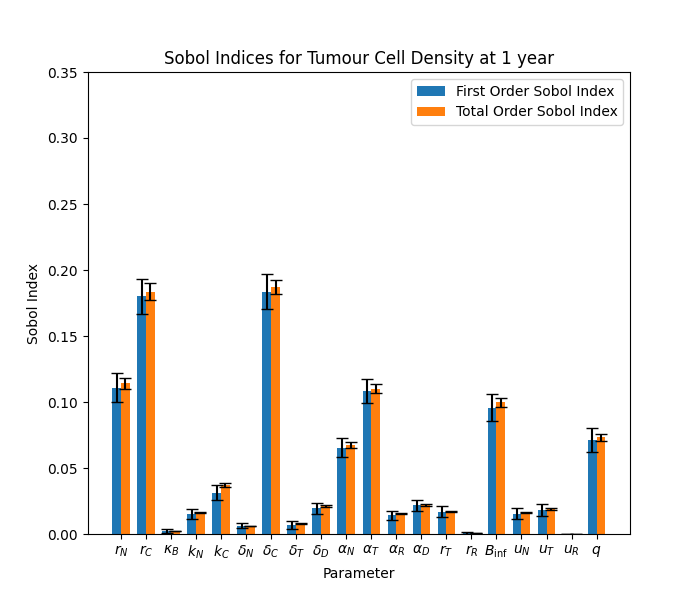}
    \includegraphics[width=0.49\linewidth]{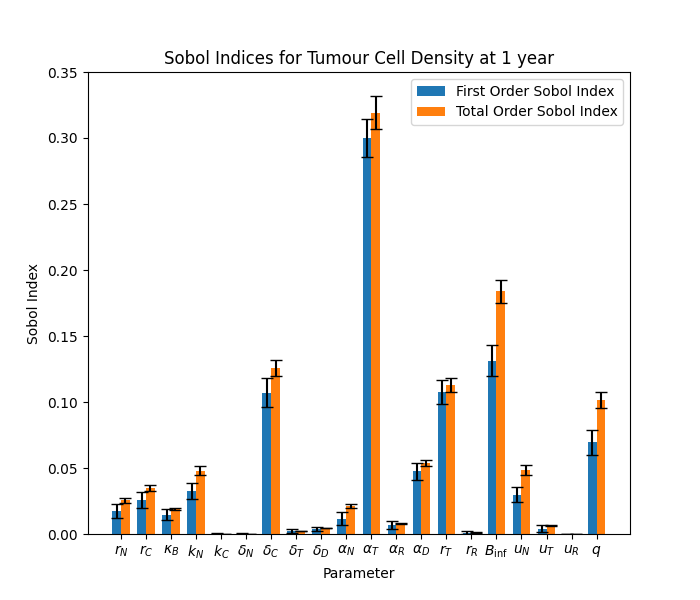}
    \caption[First and total order Sobol indices of model parameters]{First and total order Sobol indices for model parameters on cancer cell density at one year. Parameters were varied uniformly within $\pm15\%$ of their default value. Default values of $B_\text{inf}$ were taken as 10 mM (left) and 30 mM (right). Error bars indicate $95\%$ confidence intervals on indices. 180,224 simulations were run for each plot, each with a different parameter set.}
    \label{fig: SSA}
\end{figure}

We see that the relative sizes of Sobol indices are comparable to the local sensitivity coefficients. This indicates that local and global behaviour are consistent, meaning there are likely no bifurcations in this sampling range that affect parameter sensitivities. Both methods show that for low butyrate levels, cancer density is most sensitive to the intrinsic growth rate ($r_C$) and the killing rate of T cells ($\delta_C$). For high butyrate availability, $B_\text{inf}$ itself as well as parameters describing T cell dynamics, namely T cell death rate ($\alpha_T$), T cell killing rate ($\delta_C$) and T cell proliferation rate ($r_T$) are most influential. 

Results from the sensitivity analyses motivate simulation experiments carried out in the proceeding section. While $r_C$, $r_T$ and $\alpha_T$ are inherent properties of cancer and T cell population dynamics and cannot generally be directly modified in a patient, certain treatments can alter $B_\text{inf}$ and $\delta_C$. Looking at the effect of these treatments is the aim of the next section.

\section{Results}
In this section, we present simulation results obtained from coupling the gut model (\ref{eq: gut l})-(\ref{eq: gut m}) and the cancer model (\ref{eq: N})-\ref{eq: B}).

\subsection{Increased fiber intake enhances butyrate levels in the TME, helping to support normal cells and suppress cancer cells.}

We first investigate the effect of changing the fiber inflow concentration, $I_\text{inf}$, on cell populations. Figure \ref{fig: Iinf separate vars} shows time series for cell populations and butyrate in the TME under low, moderate and high fiber concentrations ($I_\text{inf}=10,20,30$ g/L respectively). We observe that moderate and high values of $I_\text{inf}$ help to decrease the deterioration of healthy colonocytes (a) while causing significant reduction in cancer cell density over time (b). Figures (c), (d) and (e) indicate a drop in immune cell densities as $I_\text{inf}$ is increased, which is consistent with the biological expectation that immune cell density should decrease in response to a decrease in cancer. As expected, butyrate concentration in the TME increases with $I_\text{inf}$ (f).

Figure \ref{fig: Iinf SS} shows the effect of varying $I_\text{inf}$ continuously between 0 and 100 g/L on steady states of the cell populations, suggesting that increasing fiber intake can cause cancer suppression, leading to a decrease in immune cells as well.

\begin{figure}
    \centering    
    \includegraphics[width=1.0\linewidth]{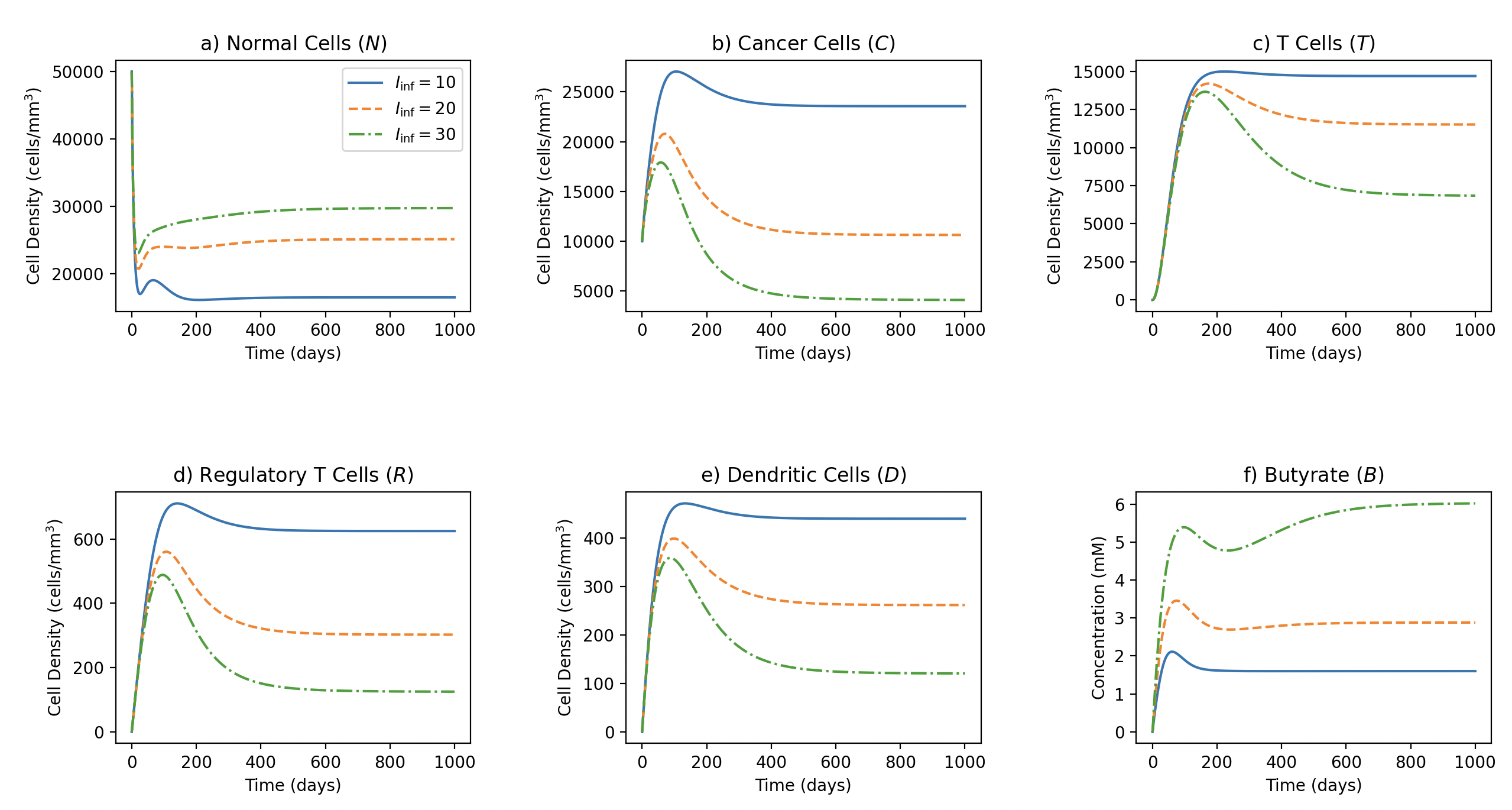}
    \caption[Solutions to the combined model plotted over time with varying amounts of fiber intake]{Solution components representing normal cells (a), cancer cells (b), T cells (c), Tregs (d), dendritic cells (e), and butyrate (f) until steady state for $I_\text{inf}=$ 10, 20 and 30 g/L.}
    \label{fig: Iinf separate vars}
\end{figure}

\begin{figure}
    \centering
    \includegraphics[width=1.0\linewidth]{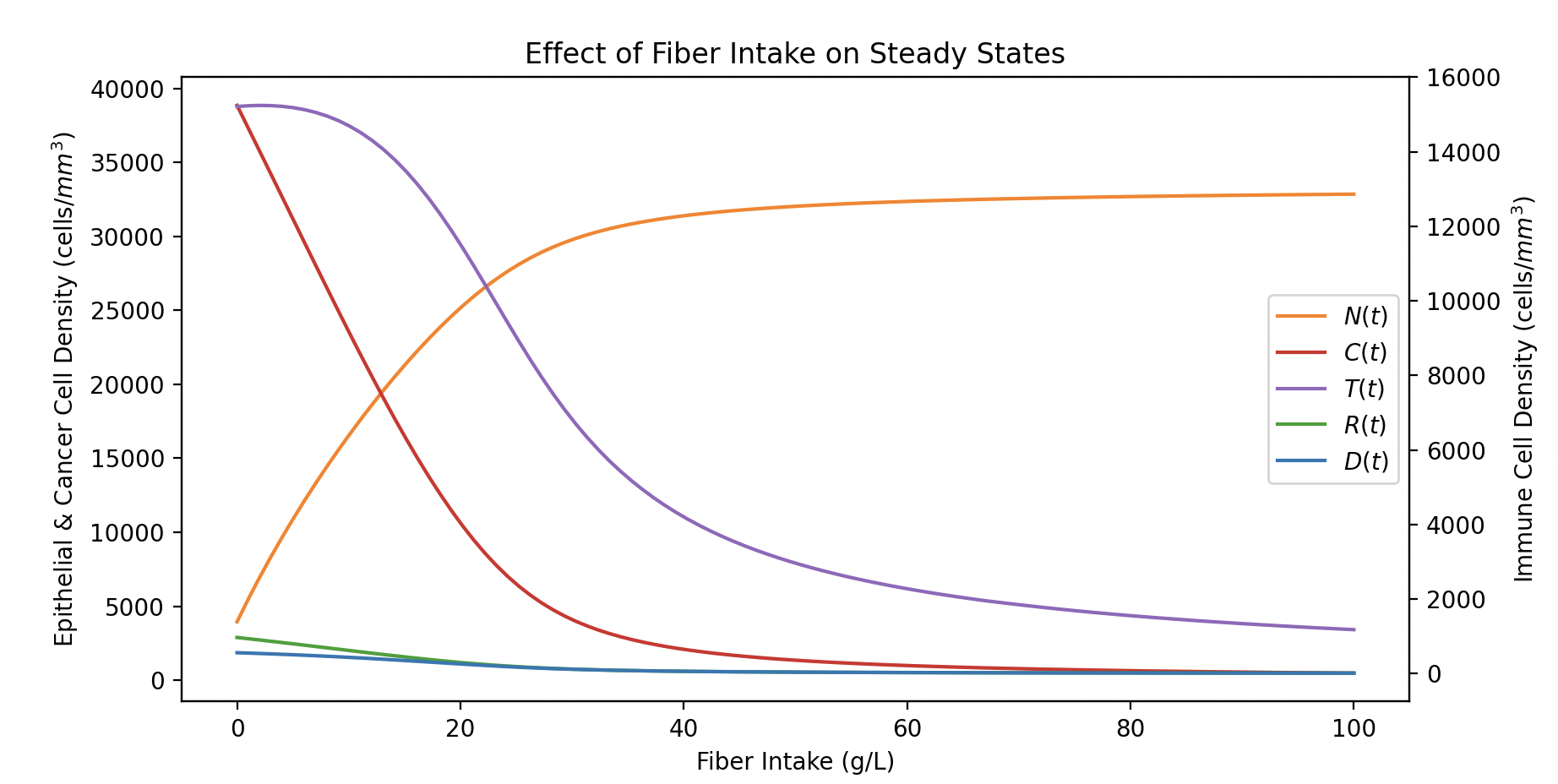}
    \caption[Effect of fiber inflow on cell populations at steady state]{Steady states of cell populations while $I_\text{inf}$ is varied continuously up to 100 g/L.}
    \label{fig: Iinf SS}
\end{figure}

\subsection{Probiotic administration can help suppress cancer growth under dysbiotic conditions in the gut.}
Probiotics are incorporated in the model by adding a constant source term $u$ (with units of g/(L$\cdot$day)) to the equation for sugar degrading biomass in the lumen. When the gut microbiome is well populated with sugar degraders, the addition of probiotics into the system has no significant effect on the amount of butyrate produced during fiber fermentation and thus little to no effect on normal and cancerous cell densities. However, when the gut microbiome is deficient of sugar degraders, the administration of probiotics is essential for butyrate production, regardless of the amount of fiber present. Figure \ref{fig: vary u} shows time series of normal and cancerous cell densities and butyrate concentration in the TME in an environment with low sugar degraders with and without probiotic administration. This result suggests that implementation of probiotics with sufficient fiber intake may be beneficial for boosting butyrate production and slowing cancer growth if the gut microbiome is deprived of sugar degrading biomass.

\begin{figure}
    \centering
    \includegraphics[width=1.0\linewidth]{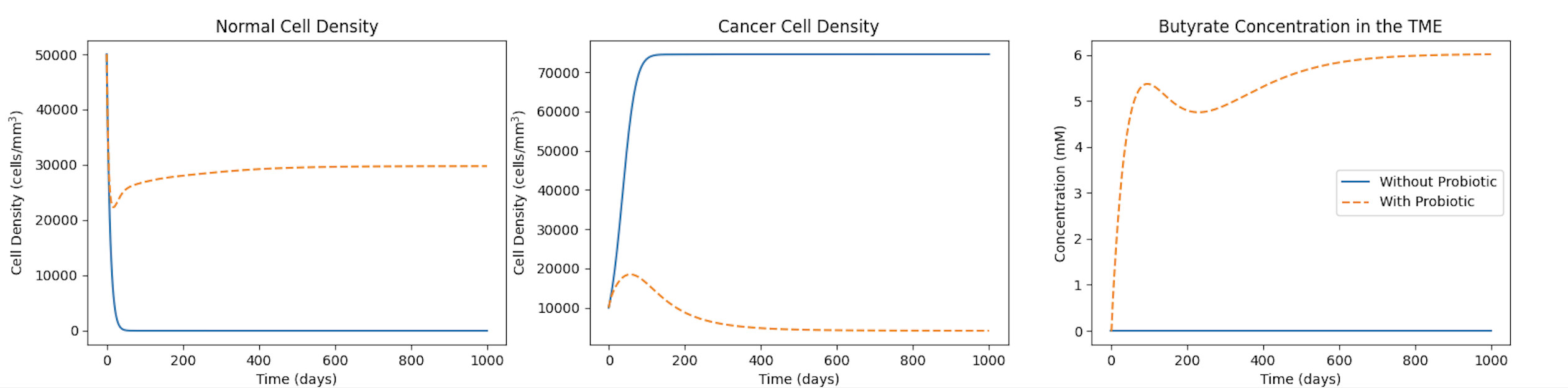}
    \caption[Effect of probiotic administration in the case of gut dysbiosis]{Solution components to the coupled gut-cancer model representing normal (left) and cancerous (middle) cell density as well as butyrate concentration (right) in the TME with and without probiotic treatment. Treatment with probiotics was implemented into the lumen at a constant rate of $u=0.1$ g/(L$\cdot$day). Initial conditions of $X_{1,l}(0)=X_{1,m}(0)=0.01$ g/L were used to simulate a microbiome environment with low concentrations of sugar degraders. $I_\text{inf}$ was set to 30.0 g/L.}
    \label{fig: vary u}
\end{figure}

\subsection{Increased fiber intake can enhance the efficacy of immunotherapy against CRC.}

Immunotherapy is a targeted treatment that helps strengthen the immune system's response to cancer. It has transformed treatment strategies against many types of cancers but immunotherapy efficacy against CRC is currently limited \cite{brahmer2012safety}. Tumours with high microsatellite instability (MSI) are often treated with Programmed Cell Death Protein 1 (PD-1) inhibitors. PD-1 is a protein on T cells that, under normal conditions, help prevent them from attacking other cells in the body. However, only 10\%-15\% of total CRC cases have high MSI, with the rest being microsatellite stable (MSS) \cite{boland2010microsatellite}. Microsatellites are short, repetitive sequences of DNA distributed throughout the genome. These regions are particularly susceptible to errors during DNA replication, most of which are corrected by the DNA mismatch repair (MMR) system. However, when the MMR system is impaired, these errors accumulate and can lead to alterations in microsatellite sequences, causing MSI. Butyrate has been shown to improve efficacy of PD-1 inhibitors in both MSI-high and MSS tumours, indicating it could be a promising treatment to CRC patients who are resistant to PD-1 inhibitors \cite{danne2021butyrate, kang2023roseburia}.

We model the administration of a PD-1 inhibitor by increasing $\delta_C$, which is the model parameter describing T cell response against cancer. Figure \ref{fig: Immunotherapy w B} shows effects of increasing $\delta_C$ beyond its default value of $5\times10^{-6}$ mm$^3$/(d$\cdot$cell) on steady states, with different values of $B_\text{inf}$. Little effect is seen in the long term behaviour of normal cells and butyrate as $\delta_C$ is varied, but a significant drop in the steady states of cancer cell density is observed, which is further amplified by high $B_\text{inf}$ values. Consequently, steady states of immune cell densities follow a similar pattern. Thus, the model suggests that adding butyrate into the system can significantly improve the efficacy of PD-1 inhibitors, which is consistent with experimental results from \cite{kang2023roseburia}. 

\begin{figure}
    \centering
    \includegraphics[width=1.0\linewidth]{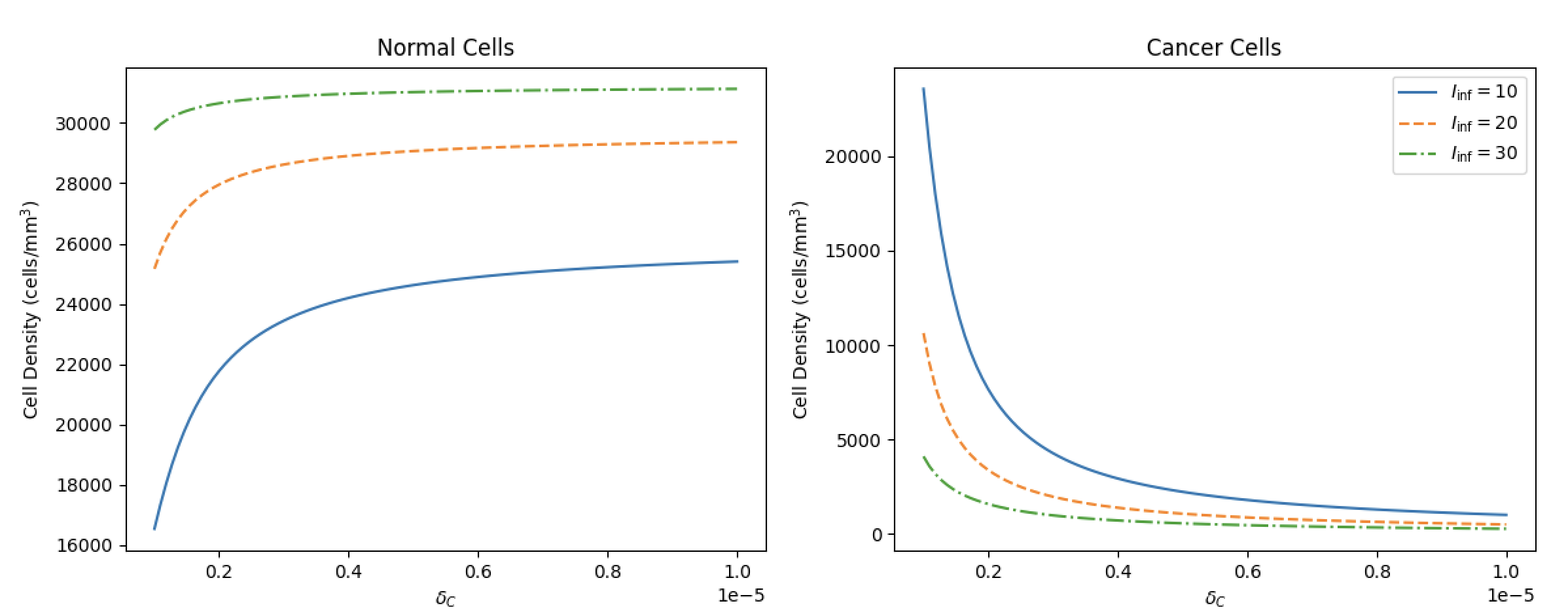}
    \caption[Cell populations at steady state with varying amounts of butyrate and PD-1 inhibitor administration] {Steady states of dependent variables while $\delta_C$ is varied continuously between $5\times10^{-6}$ and $5\times10^{-5}$ mm$^3$/(day$\cdot$cell) for $B_\text{inf} =$ 0, 10, 20 and 30 mM.}
    \label{fig: Immunotherapy w B}
\end{figure}

\subsection{Virtual Clinical Trials}
The use of virtual clinical trials (VCTs) is a popular in-silico technique where simulation experiments are carried out on a virtual patient cohort obtained by sampling model parameters from physiologically plausible ranges \cite{craig2023practical}. This helps to accommodate for biological variability between individuals, making theoretical models valuable tools for enhancing personalized medicine. We use VCTs to explore how different treatment dosages of butyrate and prebiotics affect tumour density. 

Figure \ref{fig: VPs} shows results from running VCTs using our combined model to quantify the effect of $I_\text{inf}$ on tumour density from a sample of 1000 virtual patients, each represented by a unique parameter set. Parameter sets were generated by sampling uniformly within $\pm 15\%$ of each default parameter value. Three treatment groups plus a control group were considered for both butyrate and prebiotic administration. The same patients were used for each treatment and were categorized by relative tumour densities one year into treatment compared to densities at the time of treatment initiation. Virtual patients who didn't receive fiber all had a greater than 50\% increase in tumour density after one year. Patients receiving 15 g/L and 30 g/L of fiber were split between all four categories, whereas those receiving 45g/L all experienced tumour decrease.

\begin{figure}
    \centering
    \includegraphics[width=0.7\linewidth]{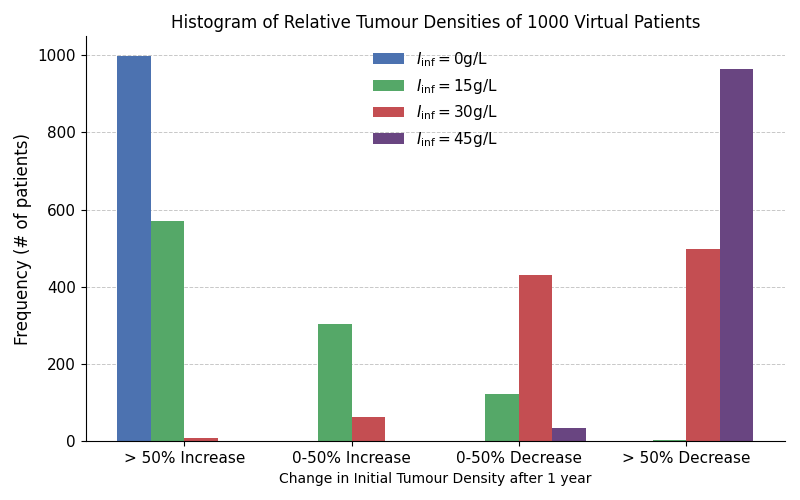}
    \caption[Tumour response from butyrate and prebiotic administration in virtual patients]{Frequency of tumour densities for 1000 generated virtual patients after one year of treatment with different concentrations of fiber intake per day (right). }
    \label{fig: VPs}
\end{figure}

\section{Stochastic Cancer Dynamics}
Cancer growth is influenced by random biological fluctuations such as genetic mutations, changes to the microenvironment, and complex immune system responses, which result in inherent uncertainty. To help account for this variability, a typical approach is to add ``noise" to an underlying cancer growth equation, thereby converting the ODE to a stochastic differential equation (SDE). As is standard in modelling cancer dynamics (see \cite{mansour2022stochastic} for example), we add noise to equation (\ref{eq: C}) that scales with cancer density, converting the equation into an SDE. 

Let $(\Omega,\mathcal{F},\{\mathcal{F}_t\}_{t\geq0},\mathbb{P})$ be a complete probability space with filtration $\{\mathcal{F}_t\}_{t\geq0}$ satisfying the usual conditions (i.e.\ it is increasing, right-continuous, and $\mathcal{F}_0$ contains all $\mathbb{P}$-null sets). 
Adding  noise to Equation (\ref{eq: C}) in the cancer model yields the SDE system
    \begin{align}        
    dX(t)=b(X(t))\,dt+\sigma(X(t))\,dW(t),\label{sde: full system}
    \end{align}
    where $X(t)=(N(t),C(t),T(t),R(t),D(t),B(t))^T\in\mathbb{R}^6$, $b:\mathbb{R}^6\to\mathbb{R}^6$ with $b(X(t))$ representing the vector of right hand sides of equations (\ref{eq: N})--(\ref{eq: B}), $\sigma:\mathbb{R}^6\to\mathbb{R}^{6\times6}$ such that $\sigma(X(t))=\Sigma$ where
    \begin{align*}
        \Sigma_{ij}=
        \begin{cases}
            \sigma_C\,X_i&\text{ if }i=j=2,\\
            0&\text{ otherwise,}
        \end{cases}
    \end{align*}
    and $\sigma_C>0$ (with dimensions $t^{-1/2}$) is the volatility constant, describing the intensity of the noise. The vector $W(t)$ is a six-dimensional standard Brownian motion but note that only its second component, which we call $W_C(t)$, survives the $\Sigma$ transformation. 
\subsection{Existence and Uniqueness}
\begin{theorem}
    For any $T>0$, the system (\ref{sde: full system}) subject to non-negative constant initial data has a unique solution that remains in $\mathbb{R}_{+,0}^6$ up to time $T$. That is, given $X(0)=X_0\in\mathbb{R}^6_{+,0}$,
    \begin{align*}
        \mathbb{P}(X(t)\in\mathbb{R}^6_{+,0})=1\hspace{0.5cm}\text{for all}\hspace{0.5cm}t\in[0,T].
    \end{align*}
\end{theorem}

\begin{proof}
    We begin by showing that $X(t)$ remains in $\mathbb{R}^6_{+,0}$ given $X_0\in\mathbb{R}^6_{+,0}$. Starting with $X_2\equiv C$, we have
    \begin{align}
        dC(t)=b_2(C(t),B(t),T(t))\,dt+\sigma_C\,C(t)\,dW_C(t).\label{eq: dC}
    \end{align}
    where $b_2$ is the second component of $b$ (ie. $b_2(C(t),B(t),T(t))$ is the right hand side of equation (\ref{eq: C})). We define the integrating factor $U(t)=e^{V(t)}$ where
    \begin{align*}
        V(t)=\tfrac{1}{2}\sigma_C^2t-\sigma_CW_C(t).
    \end{align*}
    Then, $V(t)$ satisfies
    \begin{align}
        dV(t)=\tfrac{1}{2}\sigma_C^2\,dt-\sigma_C\,dW_C(t). \label{eq: dV}
    \end{align}
    By It\^{o}'s formula (see \cite{oksendal2003stochastic} for example) applied to $f(x)=e^x$,
    \begin{align}
     d\,f(V(t))=dU(t)&=e^{V(t)}dV(t)+\tfrac{1}{2}e^{V(t)}dV(t)^2\nonumber \\
        &=\sigma_C^2\,U(t)\,dt-\sigma_CU(t)\,dW_C(t) ,\label{eq: dU}    
    \end{align}
    where the second inequality comes from plugging in the right hand side of (\ref{eq: dV}) and using that $dt^2=dt\cdot dW(t)=0$ and $dW(t)^2=dt^2$. Then by It\^{o}'s product rule,
    \begin{align*}        d(U(t)C(t))&=C(t)\,dU(t)+U(t)\,dC(t)+dU(t)\,dC(t)\\
    &=U(t)\,b_2(C(t),B(t),T(t))\,dt
    \end{align*} 
    where the second equality comes from plugging in the right hand sides of equations {(\ref{eq: dC})} and {(\ref{eq: dU})} for $dC(t)$ and $dU(t)$, respectively. Then, by setting $Y(t)=U(t)\,C(t)$, we obtain the ODE
    \begin{align*}
        dY(t)=U(t)\,b_2\left(\tfrac{Y(t)}{U(t)},B(t),T(t)\right)\,dt,
    \end{align*}
    for which $Y(t)$ is a solution for almost every $\omega\in\Omega$:
    \begin{align}
        \frac{dY(t)}{dt}&=U(t)\,b_2\left(\tfrac{Y(t)}{U(t)},B(t),T(t)\right)\nonumber\\
        &=r_CY(t)\left(1-\tfrac{C(t)}{k_C}\right)I(B(t))-\delta_CY(t)T(t).\label{sde: Y}
    \end{align}
    This preserves non-negativity of $Y(t)$ given $Y(0)=Y_0\geq0$. Thus, for all $t>0$, since $U(t)>0$, we have $C(t)\geq0$ for almost every $\omega\in\Omega$. Non-negativity of $N(t)$, $T(t)$, $R(t)$, $D(t)$ and $B(t)$ then follows by the argument used for the deterministic system (see the proof of Theorem \ref{thm: det existence}).  

    The coefficients $b(x)$ and $\sigma(x)$ are locally Lipschitz continuous for $x\in\mathbb{R}^6_{+,0}$. Therefore, by the standard theorem for local existence and uniqueness for SDEs (see for example \cite{arnold1974stochastic,friedman1975stochastic}), there is a unique solution $X(t)$ to (\ref{sde: full system}) for $t\in[0,\tau_e)$ where $\tau_e$ is the explosion time. The standard theorem for global existence and uniqueness (see \cite{oksendal2003stochastic} for example) cannot be used in this case, since $b$ and $\sigma$ do not satisfy the required linear growth condition. Instead, we define a Lyapunov function and use Khasminskii’s theorem for almost sure non-explosion \cite{khasminskii2011stochastic}. We
    define the Lyapunov function $V:\mathbb{R}_{+,0}^6\to\mathbb{R}$ as
    \begin{align*}
        V(x)=1+|x|,
    \end{align*}
    where $|x|:=\sum_{i=1}^6x_i$ for $x\in\mathbb{R}^6$. Then clearly $V(x)\geq0$ for $x\in\mathbb{R}^6_{+,0}$, $V(x)\in C^2(\mathbb{R}^6)$, and $V(x)\to\infty$ as $|x|\to\infty$. The generator $\mathscr{L}$ (acting on $V$) for (\ref{sde: full system}) is defined by
    \begin{align*}
        \mathscr{L}V(x)=\sum_{i=1}^6b_i(x)\,\frac{\partial V(x)}{\partial x_i}+\frac{1}{2}\sum_{i,j=1}^6a_{i,j}(x)\,\frac{\partial^2 V(x)}{\partial x_i\partial x_j},
    \end{align*}
    where $a:=\sigma\sigma^T$. Thus, 
    \begin{align*}
        \mathscr{L}V(x)=\sum_{i=1}^6b_i(x)+\frac{1}{2}\sigma_C^2\,\frac{\partial^2 V}{\partial x_2^2}=\sum_{i=1}^6b_i(x),
    \end{align*}
    since the second partial derivative of $V$ is zero. Then, 
    \begin{align*}
        \sum_{i=1}^6b_i(x)&\leq r_Nx_1+(r_C+\varphi_R+\varphi_D)x_2+r_Tx_3+r_Rx_4+\varphi_Tx_5+q_BB_\text{inf}\leq K(|x|+1),
    \end{align*}
    where $K=\max\{r_N,r_C+\varphi_R+\varphi_T,r_T,r_R,\varphi_T,q_BB_\text{inf}\}$. Therefore, we have
    \begin{align*}
        \mathscr{L}V(x)\leq KV(x).
    \end{align*}
    By Khasminskii's theorem, $P(\tau_e=\infty)=1$ implying that for any initial condition $X_0\in\mathbb{R}^6_{+,0}$, there is a unique, global solution to (\ref{sde: full system}).\\
\end{proof}

\subsection{Noise-Induced Tumour Extinction}\label{sec: NITE}
For the deterministic system (\ref{eq: N})--(\ref{eq: B}), we found that with default parameter values, cancer cell density persists over time. However, in the stochastic system, cancer elimination as $t\to\infty$ is guaranteed when the noise intensity is large enough. We summarize this result in the following theorem:
\begin{theorem}
    If $\sigma_C>\sqrt{2r_C}$ and $X_0\in\mathbb{R}^6_{+,0}$, then the cancer will eventually almost surely die out. That is,
    \begin{align}            \mathbb{P}\left(\lim_{t\to\infty}C(t)=0\right)=1.\label{eq: extinction result}
    \end{align}
\end{theorem}

\begin{proof}
    The proof is adapted from \cite{liu2015sufficient} where this result is shown for the stochastic generalized logistic equation. Applying It\^{o}'s formula to $f(x)=\ln{x}$ yields
    \begin{align*}
        d\ln{C(t)}=\left[r_C\left(1-\frac{C(t)}{k_C}\right)I(B(t))-\delta_CT(t)-\frac{\sigma_C^2}{2}\right]dt+\sigma_C\, dW(t).
    \end{align*}
   
    \noindent Dividing both sides by $t>0$ and converting to integral form gives
    \begin{align*}
        \frac{\ln{C(t)}}{t}=\frac{\ln{C(0)}}{t}+\frac{1}{t}\int_0^t\left[r_C\left(1-\frac{C(s)}{k_C}\right)I(B(s))-\delta_CT(s)-\frac{\sigma_C^2}{2}\right]ds+\frac{\sigma_C}{t}\int_0^tdW(s).
    \end{align*}
    By taking $t\to\infty$ on both sides, we obtain
    \begin{align*}
        \limsup_{t\to\infty}\frac{\ln{C(t)}}{t}&= \limsup_{t\to\infty}\frac{1}{t}\int_0^t\left[r_C\left(1-\frac{C(s)}{k_C}\right)I(B(s))-\delta_CT(s)-\frac{\sigma_C^2}{2}\right]ds\\
        &\leq\limsup_{t\to\infty}\frac{1}{t}\left[\int_0^tr_C\,ds-\frac{\sigma_C^2}{2}t\right]\\
        &=r_C-\frac{\sigma_C^2}{2}
    \end{align*}
    where we used the fact that the stochastic integral evaluates as $\frac{1}{t}\int_0^tdW(s)=\frac{W(t)}{t}$ and thus by the strong law of large numbers (see \cite{mao2006stochastic} for example), 
    \begin{align*}
        \lim_{t\to\infty}\frac{W(t)}{t}=0\hspace{0.2cm}a.s.
    \end{align*}
    We also used the property that $C(t)$, $B(t)$ and $T(t)$ are non-negative and $0\leq I(B(t))\leq1$ for all $t>0$. Thus, we have
    \begin{align*}        &\limsup_{t\to\infty}\frac{\ln(C(t))}{t}\leq r_C-\frac{\sigma_C^2}{2}\hspace{0.2cm}a.s.
    \end{align*}
    By the definition of the limit superior, for all $\varepsilon>0$, there is some time $\tau$ such that for all $t>\tau$,
    \begin{align*}
        \frac{\ln(C(t))}{t}\leq -\alpha+\varepsilon,
    \end{align*}
    where $\alpha:=\frac{\sigma_C^2}{2}-r_C>0$ assuming $\sigma_C>\sqrt{2r_C}$. Setting $\varepsilon=\frac{\alpha}{2}$ then yields
    \begin{align*}
        C(t)\leq e^{-\frac{\alpha}{2}t},
    \end{align*}
    and (\ref{eq: extinction result}) follows.\\
\end{proof}
Sample solution paths were generated using {\tt sdeint}, the Python library for numerical integration of It\^{o} and Stratonovich SDEs \cite{sdeint}. Figure \ref{fig: SDE mean solns} shows the mean cancer cell density over time from 500 sample paths for $\sigma_C<\sqrt{2r_C}$ (left) and $\sigma_C>\sqrt{2r_C}$ (right), demonstrating the suppressive effect of high noise intensity. Figure \ref{fig: SDE vary sigma} shows the effect of continuously varying $\sigma_C$ relative to $\sqrt{2r_C}$ on the mean normal and cancer cell densities at day 1000 over 100 sample paths for each $\sigma_C$ value with different values of $B_\text{inf}$. For both low and high $B_\text{inf}$, we see the mean cancer cell densities dying out, consistent with our theoretical result. Additionally, the plots suggest that higher butyrate levels decrease the critical noise intensity required to suppress the cancer. 

\begin{figure}
    \centering
    \includegraphics[width=0.49\linewidth]{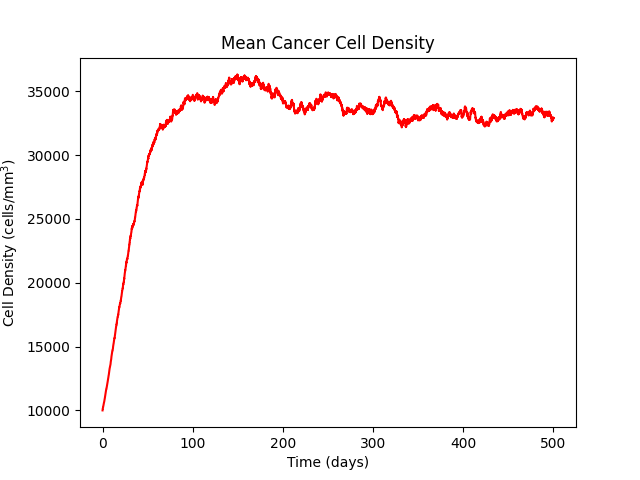}
    \includegraphics[width=0.49\linewidth]{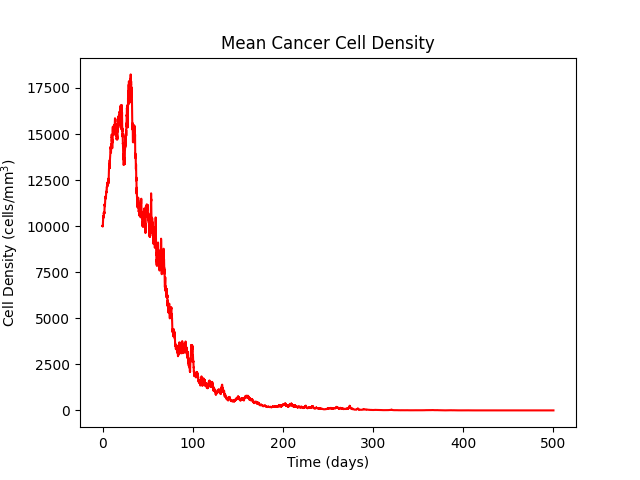}
    \caption[Mean cancer cell density over time]{Mean cancer cell density over time taken from 500 sample paths with $\sigma_C=0.1$ days$^{-1/2}$ (left) and $\sigma_C=0.4$ days$^{-1/2}$ (right).}
    \label{fig: SDE mean solns}
\end{figure}

\begin{figure}
    \centering
    \includegraphics[width=0.49\linewidth]{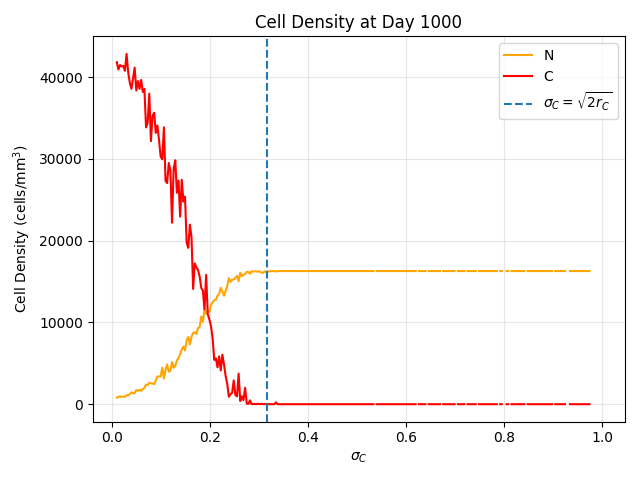}
    \includegraphics[width=0.49\linewidth]{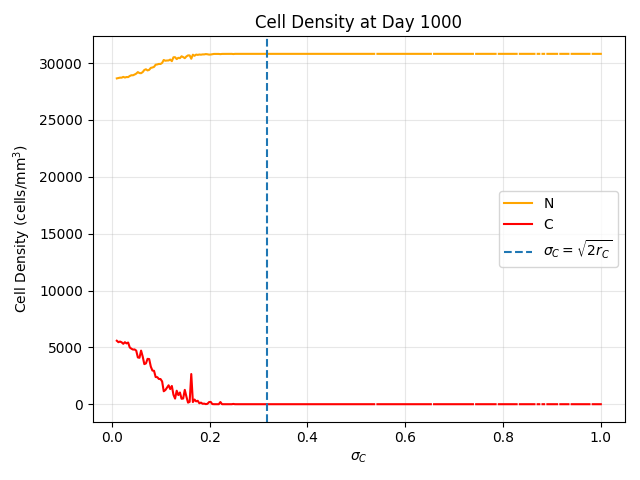}
    \caption[Effect of noise intensity and butyrate availability on cell densities]{Mean density of normal and cancer cells at day 1000 taken from 100 sample paths as $\sigma_C$ is varied continuously up to 1.0 days$^{-1/2}$  with $B_\text{inf}=5$ mM (left) and $B_\text{inf}=20$ mM (right). The line $\sigma_C=\sqrt{2r_C}$ is shown for reference.}
    \label{fig: SDE vary sigma}
\end{figure}

\section{Discussion}

\subsection{Summary of Findings}

We constructed a model describing interactions between normal colonocytes, colon cancer cells, immune cells and butyrate in the TME. While many mathematical models have been developed that incorporate interactions between cancer and immune cells, few include the effects of gut microbiota \cite{hadjigeorgiou2025mathematical}, despite growing evidence of their significance in CRC. In particular, butyrate produced from fiber fermenting bacteria has been shown to slow cancer growth through metabolic and immunological mechanisms \cite{li2018butyrate, sun2024butyrate}. We include these dynamics in our model by incorporating butyrate from the gut model, and three key types of immune cells - T cells, regulatory T cells and dendritic cells. In low butyrate concentrations, the general behaviour of cell densities in our model is consistent with established cancer-immune models and biological literature. We find that normal and cancer cells compete with each other - only one of which can persist at high densities as seen in \cite{alharbi2020new}. Similar to findings from models developed in \cite{kirshtein2020data} and \cite{raeisi2024mathematical}, our model shows that T cells and dendritic cells grow over time with increasing cancer cell density. Additionally, Tregs grow with cancer cell density, as found in \cite{de2013mathematical}. With default parameters and low butyrate availability, our model suggests that immune cells can not eliminate cancer, consistent with \cite{alharbi2020new}. 

We (weakly) coupled the gut and cancer models to determine the effects of fiber on cell dynamics. Increased fiber inflow results in higher butyrate production during fermentation, leading to decreased cancer cell density and increased normal cell density in our model. When accompanied by sufficient fiber intake, our model showed that probiotic treatment may also be beneficial if SCFA-producing gut bacteria are scarce, which is a common occurrence in CRC patients \cite{loftus2021bacterial}. This result is consistent with experimental findings discussed in \cite{shrifteylik2023current}, which suggest that fiber can help reduce cancer. Conversely, if there are already sufficient levels of sugar degrading biomass in the gut, our model suggests that probiotics show no additional benefit to butyrate production or cancer inhibition. The modelling work in \cite{vadivelmathematical} showed that antibiotics can cause dysbiosis by disrupting the natural microbial composition of the gut. Similarly, the lack of probiotic benefit in a healthy gut environment could be attributed to an overpopulation of certain species, causing microbial imbalance.

Finally, we considered a stochastic version of our cancer model to help account for inherent biological variability in the system. In comparison to the deterministic model, the stochastic model can give rise to very different solution behaviour. In particular, we showed that when the noise intensity is sufficiently high relative to the intrinsic cancer growth rate, the cancer will die out in the limit as $t\to\infty$ with probability one. This is consistent with the extinction result proved for the stochastic generalized logistic equation \cite{liu2015sufficient}. We also found that higher butyrate concentrations decrease the noise intensity threshold required to guarantee cancer elimination. 

\subsection{Limitations}
Limitations on biological models are inevitable due to the high complexity of the systems they describe. Here, we discuss a non-exhaustive list of the main limitations of our model.

Due to the lack of available data, our equations are based on related mathematical models and qualitative findings from the biological literature, but the specific functions could be argued differently. For example, even though the Gompertz growth model tends to fit tumour data more accurately \cite{tabassum2019mathematical}, we chose to model cancer growth logistically for the sake of simplicity. Additionally, we used standard sigmoidal saturation functions for butyrate uptake and inhibition ($g(B)$ and $I(B)$ respectively) to describe the metabolic effects of butyrate on normal versus cancer cells, but in reality, these interactions are much more complicated. We also assumed that immune interactions are linear, which although standard and used in similar models \cite{alharbi2020new, de2001mathematical, mohammad2022investigating, kirshtein2020data}, may be an oversimplification of the underlying mechanisms. 

Furthermore, simulation results are parameter-dependent and thus may only be valid for the parameter values we used. We attempted to mitigate this by performing sensitivity analysis to quantify the influence of parameters on model outputs as well as virtual clinical trials to look at results across a range of parameter sets. However, even these results only capture the dynamics arising from parameter sets that are within a given distance of our default values, which were obtained from various sources in the literature and may not be completely accurate. 

Throughout our research, we assumed spatial homogeneity, allowing for the use of ODEs rather than partial differential equations (PDEs). For the gut model, this means treating each compartment as a continuously stirred tank reactor (CSTR). Different approaches have been used to model spatial dynamics in the gut. For example, in \cite{munoz2010mathematical}, three reactors are connected in series to represent the proximal, transverse and distal parts of the colon discretely, where each compartment is treated as a CSTR. Alternatively, one-dimensional transport PDEs were used in \cite{moorthy2015spatially} to model the continuous flow of materials through the colon. Spatial variation was not included in the cancer model either, which could give rise to different results.

In \cite{fulbright2017microbiome}, the authors review the effects of certain microbial species on the Hallmarks of Cancer. For example, \textit{B. thetaiotaomicron} (a sugar degrading generalist) has been shown to improve T cell response against cancer although the mechanisms in which this occurs are unknown. There is also evidence suggesting that the microbiome - CRC relationship may be bidirectional, with cancer affecting the microbial composition of the gut \cite{duttabidirectional}. These effects were not included in our model. 

There are several alternative formulations that we could have used to introduce stochasticity into our model. We could have added noise to all the equations, as they also describe biological processes that are inevitably exposed to random fluctuations. However, we make the argument that cancer cell dynamics, which are highly dependent on random mutations, are likely to experience higher variability than host cell dynamics. For the sake of simplicity, we assumed that noise scales linearly with cancer cell density, but we could have considered other relationships. For example, in \cite{azizi2026advancing}, the stochastic effects scale with the square of the cancer cell density. Alternatively, instead of SDEs, we could have used random differential equations (RDEs), which incorporate randomness via parameters and give rise to solutions with differentiable sample paths, thus not requiring It\^{o} calculus.

\section{Conclusions}
Mathematical models are frequently used in biology to help with the qualitative understanding of complex systems. The main contribution of this work arises from the development and analysis of a mathematical model linking carbohydrate fermentation and colon cancer dynamics.

\begin{itemize}
    \item  The behaviour of our model is consistent with established experimental and theoretical findings that are independent of our underlying model assumptions. This indicates that connecting gut models and cell dynamics may be a promising approach for studying CRC.
    \item Our model predicts that dietary fiber enhances butyrate production, thus increasing healthy and decreasing cancerous cell densities. Dietary fiber may also have a beneficial impact on immunotherapy efficacy for CRC patients. In contrast, probiotic supplementation has little effect on SCFA production in a healthy gut but may provide a supportive benefit under dysbiotic conditions.
    \item  We find that by adding noise with sufficiently high intensity to cancer cell dynamics, in the limit as $t\to\infty$, cancer dies out almost surely. We also find that higher butyrate concentrations in the TME decrease this threshold noise intensity, weakening the sufficient conditions for cancer elimination.
\end{itemize}
While primarily theoretical, the findings from our model can be used as a basis for further modelling and experimental work. 

\pagebreak

\onehalfspace

\addcontentsline{toc}{chapter}{Bibliography}
\printbibliography

\end{document}